\documentclass[12pt]{amsart}
\usepackage{amsmath,amscd,amsthm,amsfonts, amssymb,amsxtra, mathrsfs,enumerate, stackrel, appendix}
\usepackage{calrsfs}
\usepackage{scalerel}
\usepackage[us,12hr]{datetime}
\usepackage[all]{xy} \SelectTips{cm}{}
\usepackage{hyperref}
 \hypersetup{colorlinks=true,citecolor=blue}
   \usepackage{graphicx}
   \usepackage[font=scriptsize]{caption}
   
\usepackage{todonotes}
\usepackage{tikz-cd}

\CDat
\newtheorem{thm}{Theorem}[section]  
\newtheorem*{un-no-thm}{Theorem}
\newtheorem{cor}[thm]{Corollary}     % Numbered along with thm
\newtheorem{lem}[thm]{Lemma}         % Numbered along with thm
\newtheorem{prop}[thm]{Proposition}  
\newtheorem{add}[thm]{Addendum}
\newtheorem{conjecture}[thm]{Conjecture}
\newtheorem{bigthm}{Theorem}
\newtheorem{bigcor}[bigthm]{Corollary}

\newtheorem{bigconj}[bigthm]{Conjecture}

\theoremstyle{definition}
\newtheorem{defn}[thm]{Definition}   % Numbered along with thm

\theoremstyle{definition}
\theoremstyle{definition}
\theoremstyle{remark}
\newtheorem{rem}[thm]{Remark}

\newtheorem{construct}[thm]{Construction}

\newtheorem*{out}{Outline}

\newtheorem*{intro-rem}{Remark}
\newtheorem*{intro-rems}{Remarks}

\newtheorem{ex}[thm]{Example}

 \newcommand{\ubar}[1]{\text{\b{$#1$}}}

\DeclareMathOperator{\hPic}{hPic}

\DeclareMathOperator{\Diff}{Mfd}

\DeclareMathOperator{\Pic}{Pic}

\DeclareMathOperator{\op}{op}

\DeclareMathOperator{\set}{Set}

\DeclareMathOperator{\Prin}{Bun}

\begin{document}
\title{The geometry of Floquet states}
%\date{\today, \currenttime} 

\author{Vladimir Y. Chernyak}
\address{Wayne State University,
Detroit, MI 48202} 
\email{chernyak@chem.wayne.edu} 
\author{John R. Klein} 
\address{Wayne State University,
Detroit, MI 48202} 
\email{klein@math.wayne.edu} 

\thanks{}
\begin{abstract} 
This paper studies the moduli space of ``good'' periodic Hamiltonians, i.e., periodic Hamiltonians with nondegenerate Floquet spectrum. 
The corresponding Floquet states determine a geometric model for the moduli space of such 
Hamiltonians. The model is expressed in terms of configuration spaces, flag manifolds, and principal bundles. 
A corollary of the main result is the complete determination of the rational cohomology
ring of the space of good periodic Hamiltonians. 

The paper also considers the adiabatic driving of a 1-parameter family of good periodic Hamiltonians.
We formulate two conjectures describing the resulting asymptotic quantum evolution. These are Floquet analogues of
the adiabatic theorem in quantum mechanics.
\end{abstract}

\maketitle%\pagestyle{fancy}
\setlength{\parindent}{15pt}
\setlength{\parskip}{1pt plus 0pt minus 1pt}
\def\smsh{\wedge}
\def\flush{\flushpar}
\def\dbslash{/\!\! /}
\def\:{\colon\!}
\def\Bbb{\mathbb}
\def\bold{\mathbf}
\def\cal{\mathcal}
\def\End{\text{\rm End}}
\def\stableto{\mapstochar \!\!\to}
\def\lr{{\ell r}}
\def\ad{\text{\rm{ad}}}
\def\tr{\sigma}
\def\Top{\text{\rm Top}}
\def\smallbullet{\raisebox{-1pt}{\scaleto{\!\!\bullet}{2pt}}}

\setcounter{tocdepth}{1}
{
  \hypersetup{linkcolor=olive}
  \tableofcontents
}

%\endtableofcontents

\section{Introduction} \label{sec:intro}

The classical adiabatic theorem of quantum mechanics \cite{Kato1950} describes the evolution of an eigenstate when a Hamiltonian is varied slowly. Let $H(s)$, $s\in[0,1]$, be a smooth family of self-adjoint operators on a finite-dimensional complex Hilbert space, and let $E(s)$ be a simple eigenvalue separated from the rest of the spectrum by a uniform positive gap. If a normalized state evolves under the Hamiltonian $H(t/N)$ and starts in the eigenspace for $E(0)$, then its distance from the instantaneous eigenspace for $E(t/N)$ is $O(N^{-1})$, uniformly for $0\leq t\leq N$, as $N\to\infty$. Here $t$ is physical time and $s=t/N$ is the slow parameter. More precisely, to leading order the evolution is the product of a dynamical phase and parallel transport in the eigenline bundle; for a closed path, this transport gives the Berry phase \cite{Simon1983}. This relation between slow evolution and the geometry of eigenspaces motivates the periodic analogues considered in this paper, in which Floquet states take the place of instantaneous eigenstates.

Floquet theory is the study of dynamical systems whose evolution is 
governed by differential equations with periodic coefficients. In the setting of quantum mechanics, Floquet theory arises when considering a time-dependent Schrödinger equation
\[
i\frac{d\psi(t)}{dt} =h(t)\psi(t)\, ,
\]
in which the time-dependent Hamiltonian $h(t)$ is periodic: $h(t+T) = h(t)$, where $T$ is the time period. By rescaling, we may assume that $T = 1$. 
Evolution over integer numbers of periods is governed by the monodromy operator; its eigenvalues are the Floquet quasi-periods, while the evolution within a period supplies the remaining time-dependent information. 

The corresponding adiabatic problem arises when a periodic Hamiltonian is itself varied slowly along a path of periodic Hamiltonians whose monodromy operators have simple spectrum. We formulate two conjectures concerning this evolution. The Adiabatic Floquet Conjecture (Conjecture~\ref{bigconj:adiabatic-Floquet}) predicts that the evolving state follows, in projective space, the Floquet states of the periodic Hamiltonians obtained by freezing the slow parameter. The Geometric Adiabatic Floquet Conjecture (Conjecture~\ref{bigconj:geometric-adiabatic-Floquet}) refines this approximation by identifying the phase factors and parallel transport that govern the final state. Together, these conjectures motivate the study of the geometry and topology of the space of periodic Hamiltonians with simple monodromy spectrum.

The purpose of this paper is to develop a geometric framework for these conjectures by describing Floquet states and the space of periodic Hamiltonians that admit a complete set of nondegenerate Floquet states. Our approach is based on the observation that a Floquet state is naturally a \emph{quasi-periodic} object: A solution $\psi$ satisfies
\[
\psi(t+1)=\lambda \psi(t)
\]
for some phase factor $\lambda\in U(1)$, the \emph{quasi-period}. Passing to projective space removes the phase ambiguity and produces a genuine loop. This suggests that Floquet theory should admit a formulation in terms of loop spaces, principal bundles, and holonomy.

Consider the Hopf bundle
\[
S(W)\to \mathbb P(W),
\]
where $W$ is a complex Hilbert space of finite dimension $n\geq 1$. We introduce the diffeological space
$\cal S(W)$ whose points are quasi-periodic maps $\Bbb R \to S(W)$. There is an obvious action
of $U(1)$ on $\cal S(W)$ and the quotient map 
\[
\cal S(W) \to \cal P(W)
\]
is a  principal $U(1)$-bundle. We call the points of $ \cal S(W)$ {\it normalized quasi-periodic states} and, analogously,
the points of $\cal P(W)$ are called {\it projective quasi-periodic states}. These are diffeological spaces in the sense of
\cite{IZ}. In Proposition \ref{prop:projective-quasi-periodic} below,  we exhibit a preferred diffeomorphism
\[
\mathcal P(W)\cong L\mathbb P(W)\times \Omega^1(S^1),
\]
where $L\Bbb P(W)$ is the space of smooth free loops on $\Bbb P(W)$ and  $\Omega^1(S^1)$ is the space of $1$-forms on the circle.
The factor $L\mathbb P(W)$ records the underlying projective evolution, while the $1$-form encodes the additional phase information. 
The quasi-period itself is recovered from these data by means of  a formula involving the holonomy of the canonical line bundle over projective space.

We then apply this framework to Floquet theory. For a ``good'' periodic Hamiltonian $h(t)$,
Floquet states become distinguished points in the space of quasi-periodic states.\footnote{The property of being {\it good} is a non-degeneracy condition
on the eigenvalues of the corresponding evolution operator; see Definition \ref{defn:good-periodic-hamiltonian} below.} 
We will exhibit a concrete geometric model for the diffeological moduli space $\cal H_F(W)$ of good periodic Hamiltonians and determine
its diffeomorphism type.

The rough idea is that the quasi-periods determine a configuration of points on the unit circle, while the corresponding eigenspaces determine a flag in $W$. 
These data describe a diffeological principal bundle
\begin{equation} \label{eqn:prin-bundle}
\widetilde{\mathcal H}_V(W)\longrightarrow \mathcal H_F(W)
\end{equation}
with structure group $G = T\rtimes \mathbb Z$, where $T=U(1)^{\times n}$ is an $n$-torus. We then investigate the quotient of the $G$-action
in two steps: We first quotient out by $T$ and identify the resulting space together with its $\mathbb Z$-action. Then we determine
the resulting quotient by the $\mathbb Z$-action. The main result is the following.

\begin{bigthm} \label{bigthm:diffeo-hamiltonians} There is a preferred diffeomorphism, natural under unitary isomorphisms of $W$,
\[
\cal H_F(W) \cong  \tilde C_n(S^1) \times \Omega SU(W) \times  \Omega \Bbb R \times F(W)\, ,
\]
where $\tilde C_n(S^1)$ is the universal cover of the configuration space of subsets of the circle of cardinality $n$,
$SU(W)=\ker(\det\:U(W)\to U(1))$, $\Omega SU(W)$ is its smooth based loop space,
$\Omega\Bbb R$ is the smooth based loop space of the real numbers, and $F(W)$ is the complete flag manifold on $W$.

In particular, $\cal H_F(W)$ has the structure of a nice Fréchet manifold.
\end{bigthm}

One consequence of Theorem \ref{bigthm:diffeo-hamiltonians} is that $\cal H_F(W)$ is simply connected.
Then  by Corollary \ref{cor:1-conn-classification}  below, we obtain the following.

\begin{bigcor}\label{bigcor:classification} The diffeological principal $G$-bundle \eqref{eqn:prin-bundle} is classified by an element
of the orbit set
\[
H^2(\cal H_F(W); \Lambda)/\Bbb Z\, ,
\]
where $\Lambda = \pi_1(T) = \Bbb Z^n$, and $\Bbb Z$ acts on $\Lambda$ by cyclic permutation of its coordinates.
\end{bigcor}

\begin{rem} A representative of this orbit admits a concrete description.
There is an isomorphism
\[
H^2(\cal H_F(W); \Lambda) \cong \textstyle\bigoplus\limits_{i=1}^n H^2(\cal H_F(W); \Bbb Z)\, ,
\]
and each component $H^2(\cal H_F(W); \Bbb Z)$ is isomorphic, via the first Chern class, to the group of isomorphism
classes of complex line bundles over $\cal H_F(W)$. Consequently, an element of $H^2(\cal H_F(W); \Lambda)$
amounts to specifying the isomorphism classes of an ordered $n$-tuple of complex line bundles $\eta_1,\dots,\eta_n$ over $\cal H_F(W)$.

Let $\cal H_F(W) \to C_n(S^1)$ be the map which assigns to a good periodic Hamiltonian the set
of its quasi-periods. Since $\cal H_F(W)$ is simply connected, choose a lift
$\cal H_F(W) \to \tilde C_n(S^1)$. A point of $\tilde C_n(S^1)$ may be described as an ordered $n$-tuple of real numbers
$(t_1,\dots,t_n)\in\Bbb R^n$ that satisfies the inequalities
\[
t_1 < t_2 < \cdots < t_n < t_1+1\, ,
\]
where the covering projection is given by
\[
(t_1,\dots,t_n)\mapsto\{e^{2\pi i t_1},\dots,e^{2\pi i t_n}\}\, .
\]
Thus
the chosen lift determines a total ordering of the set of quasi-periods associated with a good periodic Hamiltonian.
If $h\in\cal H_F(W)$ has $j$-th quasi-period $\lambda_j$,
then the corresponding eigenspace of its monodromy is a complex line $L_j(h)$ in $W$. The assignment $h\mapsto L_j(h)$ describes the line bundle $\eta_j$.
The orbit in Corollary~\ref{bigcor:classification} is represented by
\[
\bigl(c_1(\eta_1),\dots,c_1(\eta_n)\bigr).
\]
Changing the lift cyclically permutes these line bundles and their Chern classes, leaving the orbit unchanged.
The normalization in Construction~\ref{prop:intrinsic-product-decomposition} selects a preferred lift and hence a preferred representative of this orbit.
\end{rem}

Note that $\tilde C_n(S^1)$ and $\Omega \Bbb R$ are contractible spaces, so another 
consequence of Theorem \ref{bigthm:diffeo-hamiltonians} 
is the following result. 

\begin{bigcor}  \label{bigcor:diffeo-hamiltonians}There is a homotopy equivalence
\[
\cal H_F(W) \simeq \Omega SU(W) \times F(W)\, .
\]
In particular, the rational cohomology ring of $\cal H_F(W)$ is isomorphic to the graded tensor product
\[
\Bbb Q[s_1,s_2,\dots s_{n-1}] \otimes \Bbb Q[t_1,t_2,\dots t_{n}]/\Sigma\, ,
\]
in which $\deg s_j = 2j, \deg t_j = 2$ and $\Sigma \subset \Bbb Q[t_1,t_2,\dots t_{n}]$ is the ideal
generated by the elementary symmetric polynomials $e_1,\dots,e_n$ in $t_1,\dots,t_n$.
\end{bigcor}

\subsection{The Adiabatic Floquet Conjecture} 
The philosophy of  diffeology is that to understand an infinite dimensional space $X$
it is enough to study its plots $M \to X$, where $M$ is a finite dimensional smooth manifold.\footnote{Classically, such infinite dimensional 
spaces arose from differential topological considerations, such as a function space of smooth maps \cite{Chen_iterated}.}

Let $X = \cal H_F(W)$. Then a plot of the form
$U \times [0,1] \to \cal H_F(W)$
is called a {\it parametrized driving protocol,} in which $U$, the parameter space, is a smooth manifold.
Here we focus on the simplest case by taking $U$ to be a point.
In this case, we have a smooth 1-parameter family of good periodic Hamiltonians
\[
h\: [0,1] \to \cal H_F(W)\, .
\]
Choose a normalized Floquet state $\psi_0\in\cal S(W)$ of $h(0)$,
and fix an initial time $t_1\in\Bbb R$. Set
\[
M_s:=\mathcal U_{h(s)}(t_1+1,t_1),
\qquad L_0:=\Bbb C\psi_0(t_1).
\]
Since the monodromy operators $M_s$ have simple spectrum, the eigenline
$L_0$ continues uniquely to a smooth family of eigenlines $L_s$.
Let $\lambda(s)\in U(1)$ be the eigenvalue of $M_s$ on $L_s$.
Solving the frozen Schr\"odinger equation for $h(s)$ with any unit
initial vector in $L_s$ determines a projective quasi-periodic state
$[\psi]_s\in\cal P(W)$, independent of the phase of that vector.
This construction gives a smooth path
\[
[\psi]_\bullet\:[0,1]\longrightarrow\cal P(W),
\qquad [\psi]_0=[\psi_0].
\]
The continuation and smoothness are explained in \S\ref{sec:Floquet-conjecture}.

The pair $([\psi]_\bullet,\psi_0)$ specifies a path lifting problem for
the principal $U(1)$-bundle $p\:\cal S(W)\to\cal P(W)$.
Since the frozen equation for $h(0)$ determines $\psi_0$ from its value
$\psi_0(t_1)$, we also denote these lifting data by
\[
([\psi]_\bullet,\psi_0(t_1)).
\]

In \S\ref{sec:Floquet-conjecture},
 we construct a
  1-parameter family of connections $\nabla^F_t$ on $p$, where $t\in\Bbb R$. Let $\Gamma^F_t$ be the corresponding family of lifting functions, and let $t_0\in\Bbb R$ be given.
Let  
\[
\Gamma^F_{t_0}([\psi]_\bullet,\psi_0(t_1))\in \cal S(W) 
\]
denote the endpoint at $s=1$ of the $\nabla^F_{t_0}$-horizontal lift of $[\psi]_\bullet$ with initial state $\psi_0$, encoded by $\psi_0(t_1)$.
In particular, if $t_2\in \Bbb R$, we have $\Gamma^F_{t_0}([\psi]_\bullet,\psi_0(t_1))(t_2) \in S(W)$.

\subsubsection{The adiabatic limit} Physicists often study {\it adiabatic} (slow) processes, in part because such systems are tractable, and in part because adiabatic evolution often reveals underlying geometry. 
With this in mind, we now traverse the fixed driving protocol increasingly slowly.
Since $h$ takes values in $\cal H_F(W)$, the frozen monodromy spectrum is simple;
for $n>1$, its eigenvalues are uniformly separated by compactness of $[0,1]$.

Consider the time interval $[t_1,t_2+N]$, where $t_1,t_2\in\Bbb R$ are fixed
and $N$ is a positive integer sufficiently large that $N+t_2-t_1>0$. Define
 \[
 \tau\: [t_1,t_2+N] \to [0,1], \qquad \tau(t) = \frac{t-t_1}{N + t_2-t_1}\, .
 \]
 The {\it time rescaled Hamiltonian} $\tilde h\: [t_1,t_2+N] \to H(W)$ is
 \[
 \tilde h(t) = h(\tau(t))(t) \, .
 \]
Let $\tilde \psi\: [t_1,t_2+N] \to S(W)$ solve the Schr\"odinger equation for $\tilde h$
with respect to the initial condition $\tilde \psi(t_1) = \psi_0(t_1)$.
In the adiabatic limit, all data other than $N$ are fixed, and the dependence of
$\tau$, $\tilde h$, and $\tilde\psi$ on $N$ is understood.
Let $Q_N(t)$ and $P_s(t)$ be the orthogonal projections onto
$\Bbb C\tilde\psi(t)$ and $[\psi]_s(t)$, respectively.
Write $\|\cdot\|_{\operatorname{op}}$ for the operator norm on $W$.

\begin{bigconj}[Adiabatic Floquet Conjecture]\label{bigconj:adiabatic-Floquet}
As $N\to\infty$ one has
\[
\sup_{t\in[t_1,t_2+N]}
\bigl\|Q_N(t)-P_{\tau(t)}(t)\bigr\|_{\operatorname{op}}
=O(1/N)\, .
\]
\end{bigconj}

The authors have obtained a formal derivation of Conjecture
\ref{bigconj:geometric-adiabatic-Floquet}. A mathematically rigorous treatment
is deferred to a future paper. We now describe this refinement, which also
keeps track of phase.

In \S\ref{sec:Floquet-conjecture} we construct
auxiliary functions 
\[
\begin{gathered}
\rho\: C^\infty([0,1], U(1))\to \Bbb R/(2\pi\Bbb Z)\, ,\\
\varphi\: C^\infty([0,1], U(1)) \times \Bbb R^3 \to \Bbb R\, .
\end{gathered}
\]

\begin{bigconj}[Geometric Adiabatic Floquet Conjecture]  \label{bigconj:geometric-adiabatic-Floquet} As $N \to \infty$ one has
\[
\tilde \psi(t_2+N) \approx e^{iN\rho(\lambda)} e^{i\varphi(\lambda,t_0,t_1,t_2)} \Gamma^F_{t_0}([\psi]_\bullet,\psi_0(t_1))(t_2)\, .
\]
Here ``$\approx$'' means an $O(1/N)$ error in the norm of $W$.
\end{bigconj}

\begin{rem} A rigorous derivation of Conjecture
\ref{bigconj:geometric-adiabatic-Floquet} from Conjecture
\ref{bigconj:adiabatic-Floquet} requires further estimates on the phase
and parallel transport. We defer this argument to a subsequent paper.

 In the above, $\lambda \in C^\infty([0,1], U(1))$ is the function whose value $\lambda(s)$ is the quasi-period associated with $[\psi]_s$.
The formula provides an asymptotic description of 
$\tilde \psi(t_2+N)$ in terms of the geometric, non-adiabatic factor $ \Gamma^{F}_{t_0}([\psi]_\bullet,\psi_0(t_1))$
arising from the connection $\nabla^F_{t_0}$. The remaining terms on the right are phase factors. 
Only the first of these depends upon $N$.

It is worth remarking that although each of the last two factors on the right in the display of Conjecture \ref{bigconj:geometric-adiabatic-Floquet} depends on $t_0$, we will see that the product
$e^{i\varphi(\lambda,t_0,t_1,t_2)} \Gamma^F_{t_0}([\psi]_\bullet,\psi_0(t_1))(t_2)$ does not.
\end{rem}

\begin{out}
The paper is organized as follows. In Section \ref{sec:prelim} we consider Fréchet manifolds and diffeological spaces. We 
introduce just enough of the theory to apply it to the kinds of objects considered in the paper. In Section \ref{sec:qp-states} we introduce
the diffeological spaces of normalized and projective quasi-periodic states and identify their diffeomorphism types.
Section \ref{sec:Floquet}, which is short, defines Floquet states, their quasi-periods, and introduces the space of good periodic Hamiltonians $\cal H_F(W)$.
In Section \ref{sec:H_F-model} we produce the principal bundle \eqref{eqn:prin-bundle}. Section \ref{sec:diffeo-type-H}
proves Theorem \ref{bigthm:diffeo-hamiltonians}.  Section \ref{sec:Floquet-conjecture} formulates the Adiabatic Floquet Conjecture.

We provide two appendices.
The first of these, Appendix \ref{sec:Picard}, is about the Picard group in the setting of Hermitian line bundles over a Fréchet manifold $X$. The main point
is that the first Chern class $\Pic(X) \to H^2(X;\underline{\Bbb Z})$ is an isomorphism when $X$ is sufficiently nice. 
In  Appendix \ref{sec:prin-TtimesH}, using the material of the previous appendix,  we exhibit a cohomological classification of
principal $G$-bundles over such an $X$ when $G =  T\rtimes H$, for $T$ an $n$-torus and $H$ a discrete group.
\end{out}
\section{Preliminaries} \label{sec:prelim}

Throughout the paper we work in the category of diffeological spaces. This setting is well suited to the infinite-dimensional spaces that arise naturally in Floquet theory, such as smooth loop spaces and spaces of quasi-periodic maps, while retaining the familiar geometric features of smooth manifolds and principal bundles.

We take Hermitian inner products to be complex-linear in the first argument
and conjugate-linear in the second.

\subsection{Fr\'echet manifolds}
We refer the reader to  \cite{Kriegl-Michor} for the foundations of the theory of infinite dimensional manifolds in the convenient setting.
A {\it Fr\'echet manifold} is a $C^\infty$-manifold with the property that its model spaces are Fr\'echet spaces, 
i.e.,  locally convex, metrizable, complete topological vector spaces.

% A  Fr\'echet manifold $X$ is {\it $C^\infty$-paracompact} \cite[16.1]{Kriegl-Michor} if  for any open cover, there exists a smooth partition of 
% unity subordinate to it. Following \cite[\S11]{Kihara}, one says that $X$ is {\it hereditarily $C^\infty$-paracompact} if every open submanifold
% of $X$ is $C^\infty$-paracompact. 

\begin{defn} A Fr\'echet manifold $X$ is {\it nice} if
\begin{enumerate} 
\item  the model spaces of $X$
are nuclear Fr\'echet spaces, 
%\item $X$ is {\it semiclassical} in the sense of \cite[\S1.2]{Kihara},
\item $X$ is Hausdorff, and 
\item $X$ is paracompact.
\end{enumerate}
\end{defn}

%\begin{rem}\label{rem:Kihara-terminology} In the terminology of Kihara, every Fr\'echet manifold $X$ is  semi-classical \cite[lem.~11.5]{Kihara}.
% If we assume that $X$ is nice, then it is hereditarily $C^\infty$-paracompact  \cite[cor.~11.17]{Kihara}. In particular,
% $X$ satisfies the assumptions of  \cite[th.~1.1]{Kihara}.
%\end{rem}

\begin{ex}\label{ex:Kihara-function-space} Let $M$ and $N$ be smooth finite dimensional manifolds with $M$ compact. 
Then
\[
\mathfrak{C}^\infty(M,N)\, ,
\] the space of smooth maps $f\: M \to N$, is a nice Fr\'echet manifold \cite[\S42]{Kriegl-Michor}, 
\cite[\S11]{Kihara}, \cite{Kihara_model}. This $C^\infty$-manifold is modeled on the nuclear Fr\'echet spaces $C^\infty(f^\ast TN)$, i.e., of smooth sections
of the pullback bundle $f^\ast TN \to M$. Furthermore, by  \cite[prop.~42.3]{Kriegl-Michor}, $\mathfrak{C}^\infty(M,N)$ is Lindelöf, smoothly paracompact, has separable connected components, and is metrizable.
\end{ex}

\subsection{Diffeological spaces} A standard reference for diffeology is \cite{IZ}. 
Let $\Diff$ be the category of finite dimensional smooth manifolds in which a morphism is a smooth map.
Let 
\[
\Bbb E\subset \Diff
\] 
be the full subcategory in which an object
is an open subset of some $\Bbb R^n$. Let $X$ be a set. 
A {\it parametrization} of $X$ is a set theoretic function $U\to X$ for some object $U\in \Bbb E$.

A  {\it diffeological space} \cite{IZ} consists of a set $X$ and a specified set of
parametrizations $D_X$ satisfying the following axioms:
\begin{enumerate}[(i).]
\item Every constant function $U \to X$ is in $D_X$.
\item If $p\: U\to X$ is a parametrization, $\{U_\alpha \}$ is an open cover of $U$, and 
$\{p_{\alpha}\: U_\alpha \to X\}$ is a family of elements of $D_X$ such that $p_{\alpha}=p|_{U_{\alpha}}$,
then $p$ is in $D_X$.
\item If $p\: U \to X$ is in $D_X$ and $f\:V\to U$ is a morphism of $\Bbb E$, then 
$p\circ f\: V\to X$ is in $D_X$.
\end{enumerate}
The elements of $D_X$ are called {\it plots}.
A morphism of diffeological spaces $(X,D_X) \to (Y,D_Y)$ is a plot preserving function $f\: X\to Y$.

Henceforth, we abbreviate the notation and instead refer to $(X,D_X)$  as $X$.
The category of diffeological
spaces is denoted by $\cal D$.
\begin{rem}[Diffeological spaces as concrete sheaves]
The inclusion functor $\Diff \to \cal D$ is fully faithful.
If $M$ is a smooth manifold, we will also denote its image in $\cal D$ by $M$.

Let $\set$ be the category of sets.
Then a diffeological space $X$ defines a sheaf 
\[
[X]\: \Diff^{\op} \to \set
\]
 in which
$[X](M) = \hom_{\cal D}(M,X)$. In other words, $[X]$ is a {\it smooth stack} having no automorphisms.
%A sheaf $\Diff^{\op} \to \set$
% is a {\it smooth stack} whose objects have no automorphisms.

% The functor $X\mapsto [X]$ is a fully faithful embedding of diffeological spaces
% into the category of smooth stacks.
\end{rem}

\subsubsection{The $\cal D$-topology}
A diffeological space $X$ has an associated topological space $LX$
with underlying set $X$ whose topology is final for
its set of plots, i.e.,
a subset $A\subset X$ is open if and only if its preimage with respect to every plot $U\to X$
is open. Conversely, a topological space $Y$ determines a diffeological space $RY$
with underlying set $Y$ in which
the plots are continuous functions $U\to Y$.

Let $\Top$ be the category of topological spaces. Then we have  a pair of adjoint functors
\[
L\: \cal D \rightleftarrows \Top \: R
\]
satisfying $L\circ R\circ L = L$ and $R\circ L \circ R = R$
\cite[prop.~3.3]{Christensen-Sinnamon-Wu}. If $X$ is a diffeological space, the {\it $\cal D$-topology} on $X$ is defined to be the topology on $LX$.

\begin{rem} Let $\Delta^n$ denote the standard $n$-simplex.
A space $Y$ lies in the image of $L$ if and only if it is {\it $\Delta$-generated}, i.e., $A\subset Y$ is open if and only if 
for every continuous map $f\: \Delta^n \to Y$, with $n \in \Bbb N$, we have $f^{-1}(A) \subset \Delta^n$ is open \cite[prop.~3.10]{Christensen-Sinnamon-Wu}.
\end{rem}

%Let $X(U)$ be the set elements of $D_X$ having domain $U$.
% Then $U\mapsto X(U)$ defines a sheaf
% \[
% X({-})\: \Bbb E^{\op} \to \set\, .
% \]
% which is {\it concrete} in that it is a sub-sheaf of  $\hom_{\set}({-},X)$ (cf.~\cite{Baez-Hoffnung}, \cite[cor.~2.10]{Watts-Wolbert}).

\begin{rem} Let $\cal F$ be the category of Fr\'echet manifolds.
Then the evident functor
\[
\cal F\to \cal D 
\]
is fully faithful \cite[th.~5]{Losik}. For a Fr\'echet manifold $X$, its topology coincides with the {\it ${\cal D}$-topology,} i.e., the final topology induced by its plots.
\end{rem}

\subsection{Inductions and subductions}
Let $f\: X\to Y$ be a map of diffeological spaces.

\begin{defn}[cf.~{\cite[p.~15]{IZ}}] The map $f$ is an {\it induction} if it has the following properties:
\begin{enumerate}[(i).]
\item $f$ is one-to-one, and
\item $\phi\: U \to X$ is a plot if and only if $f\circ \phi\: U\to Y$ is a plot.
\end{enumerate}
\end{defn}
In the case of inductions $f\: X\to Y$, the diffeology on $X$ is the {\it pullback} $f^\ast D_Y$ of the diffeology
on $Y$ \cite[p.~14]{IZ}.

\begin{defn}[cf.~{\cite[p.~25]{IZ}}]  The map $f$ is a {\it subduction} if
\begin{enumerate}[(i).]
\item $f$ is onto, and
\item $\psi\: V \to Y$ is a plot if and only if for every point $y\in V$, there is a neighborhood $U$ of $y$, a plot
$\phi\: U \to X$  and a commutative diagram
\[
\xymatrix{
U \ar[r]^{\phi} \ar[d]_{\cap} & X \ar[d]^{f} \\
V \ar[r]_{\psi} &Y .
}
\]
\end{enumerate}
\end{defn}
In this case the diffeology on $Y$ is given by the {\it pushforward} $f_\ast D_X$ of the diffeology on $X$ \cite[p.~24]{IZ}.

\subsection{Diffeological principal bundles}
Let $G$ be a Lie group acting smoothly on the right of a diffeological space $P$. 
We always equip the set of orbits $P/G$ with the diffeology making $P\to P/G$
into a subduction, i.e., $P/G$ has diffeology $f_\ast D_P$.
We say that the action is {\it free} if $p\cdot g = p$ implies $g = e$.

\begin{defn}
A $G$-invariant map $\pi\: P\to X$ is a {\it diffeological principal $G$-bundle} \cite[p.~243]{IZ} if
\begin{enumerate}
\item $\pi$ is a subduction, and
\item the {\it shearing map}
\begin{align*}
P \times G & @>>> P\times_X P\\
(p,g) & \mapsto (p,p\cdot g)
\end{align*}
is a diffeomorphism, where $P\times_X P \subset P\times P$ is equipped with
the subset diffeology.
\end{enumerate}
\end{defn}

\begin{rem}
The action of $G$ on $P$ is free, and $\pi$ induces a
diffeomorphism
\[
P/G \cong X.
\]
Thus $\pi$ is identified with the quotient projection
$P\to P/G$.
\end{rem}

\subsection{Local triviality} 
Let $p\: P\to X$ be a diffeological principal $G$-bundle. One says that $p$ admits {\it local sections}
if every $x\in X$ has a $\cal D$-open neighborhood $U\subset X$ over which
$p$ admits a smooth section $U\to P$.
If this is the case, the bundle  $p^{-1}(U)\to U$ is identified with the trivial bundle $U\times G \to U$.

Suppose $p$ sits in a Cartesian square of diffeological spaces
\[
\xymatrix{
P \ar[r] \ar[d]_p & E \ar[d]^q \\
X \ar[r] & B
}
\]
in which $q\: E\to B$ is a smooth locally trivial principal $G$-bundle of finite dimensional manifolds. 
Then $q$ admits smooth local sections. Since the diagram
is cartesian, it follows that $p$ also admits local sections. It follows that $p$ is a locally trivial
 diffeological principal $G$-bundle.

\begin{rem} If $\pi\: P \to X$ is a locally trivial diffeological principal $G$-bundle where $X$ is a finite dimensional
smooth manifold, then $P$ has a smooth manifold structure and $\pi$ is a (locally trivial) smooth principal $G$-bundle.
Conversely, a  smooth principal $G$-bundle of finite dimensional manifolds is diffeological.
\end{rem}

\subsection{Double quotients} Let $G$ be a Lie group acting on the right of a diffeological space $P$.
Given a closed normal subgroup $N\subset G$, we set $H = G/N$.
Then there is a diffeomorphism
\begin{equation} \label{eqn:double-quotient-formula}
P/G \cong (P/N)/H\, .
\end{equation}
Suppose in addition that $G = N\rtimes H$ is a semi-direct product, where $N\subset G$ is a closed subgroup and $H$ is discrete.
Set $X = P/N$ and $Y = P/G$. Then we have a pair of maps
\[
P @> p >> X @> q >> Y\, .
\]

\begin{prop}\label{prop:composite-bundle} Assume that $p$ is a diffeological principal $N$-bundle and 
that $q$ is a diffeological principal $H$-bundle.
Then the composition $q\circ p\: P\to Y$ is a diffeological principal $G$-bundle. Furthermore,
if $p$ and $q$ are locally trivial, then so is $q\circ p$.
\end{prop}

\begin{proof} The map $P\to Y$ is a subduction because subductions compose \cite[1.47]{IZ}.
The $G$-action is free, since the actions of $H$ on $X$ and of $N$ on $P$ are free,
and its orbits are the fibers of $P\to Y$.
To show that the map $P\times G\to P \times P$ is an induction we use the decomposition
\[
G = \coprod_{h\in H} N_h
\]
where $N_h \subset G$ is the subspace of elements of the form $(n,h)$ for $n\in N$.
Then $N_h$ is diffeomorphic to $N$ and
\[
P \times G =  \coprod_{h\in H} P\times N_h\, .
\]
Let $\phi\:U\to P\times_Y P$ be a plot. Projecting to $X\times_Y X$ and applying
 the inverse of the shearing map for $q$ gives a smooth $H$-valued function on $U$.
Since $H$ is discrete, this function is locally constant. Thus locally $\phi$ takes values
in the image of a single $P\times N_h$, and it is enough to show for each $h\in H$ that the composite map 
\[
P\times N \cong P \times N_h \to P\times G\to P \times P
\]
is an induction. This map is a composition of the form
\[
P \times N @>s>>  P\times P @> \sigma >> P \times P
\]
where $s$ is the shearing map  for  $p$ and $\sigma$ is the diffeomorphism
given by $(a,b) \mapsto (a,b\cdot (1,h))$. Hence, we have an induction since inductions compose.

Local triviality is straightforward and is left to the reader.
\end{proof}

\section{Quasi-periodic states} \label{sec:qp-states}

\subsection{Quasi-periodic maps} Let $p\: E \to M$ be a smooth principal $U(1)$-bundle
over a finite dimensional closed manifold. Assume that $p$ is equipped with a connection.
A smooth map
\[
f\: \Bbb R \to E
\]
is {\it quasi-periodic} if there is a $\lambda \in U(1)$ such that $f(t+1)=\lambda f(t)$
for all $t\in \Bbb R$. The number $\lambda$ is called the {\it quasi-period} of $f$.

By a straightforward induction,  $f(t+n)  = \lambda^nf(t)$ for all $t\in \Bbb R$ and $n  \ge 0$.
Moreover, if we replace $t$ by $t-n$ in the latter expression, then we obtain the identity $f(t-n) = \lambda^{-n}f(t)$.
Consequently,
$f$ is $\Bbb Z$-equivariant with respect to the translation action
on $\Bbb R$ and the action on $E$ defined by $n\cdot w := \lambda^nw$.

To each quasi-periodic map, there is a uniquely defined smooth map $f^\sharp \: \Bbb R/\Bbb Z \to M$
such that the square
\begin{equation} \label{eqn:sharp}
\xymatrix{
\Bbb R \ar[r]^f\ar[d] & E \ar[d]_p \\
\Bbb R/\Bbb Z \ar[r]_{f^\sharp} & M
}
\end{equation}
commutes. Conversely, given a smooth map $g\: \Bbb R/\Bbb Z \to M$ and a choice of lift of $g(0) \in M$ to a point $e\in E$, there
is a  preferred quasi-periodic map $f\:  \Bbb R \to E$ such that $f(0) = e$. The map $f$ is defined by the connection.

Let $\ast \in E$ be a fixed choice of basepoint.

\begin{defn} 
Let
\[
Q(\Bbb R,E) \quad \text{and } \quad  Q_\ast(\Bbb R,E)
\]
denote the diffeological spaces of quasi-periodic maps $f\:\Bbb R \to E$
and, respectively, the quasi-periodic maps  such that $f(0) = \ast$.
\end{defn}

\begin{rem} A  plot $U \to Q(\Bbb R,E)$
is a smooth map $F\:U\times \Bbb R \to E$ 
such that for every $x\in U$, the function $f_x = F(x,{-})$
is quasi-periodic. We equip $Q_\ast(\Bbb R,E) \subset Q(\Bbb R,E)$ with the subdiffeology.
\end{rem}

We now specialize to the case when the $U(1)$-principal bundle is given by $S^1 \to \ast$. Let $1\in S^1$ be the basepoint.
Consider the diffeological space of smooth $1$-forms $\Omega^1(S^1)$,
 in which a plot $U \to \Omega^1(S^1)$ is a smooth section
of $\pi^*\tau^\ast_{S^1}$, where $\tau^\ast_{S^1}$ is the cotangent bundle and $\pi\: U \times S^1\to S^1$ is the projection.

\begin{lem} \label{lem:qp-forms} There are preferred diffeomorphisms
\begin{enumerate}[(i).]
\item $Q_\ast(\Bbb R,S^1) \cong \Omega^1(S^1)$, and  
\item  $Q(\Bbb R,S^1) \cong S^1 \times \Omega^1(S^1)$.
\end{enumerate}
\end{lem}

\begin{proof} We first prove (i).
Let $p\: \Bbb R \to \Bbb R/\Bbb Z = S^1$ be the projection given by $p(t) = \exp(2\pi i t)$.
For $\alpha\in  \Omega^1(S^1)$ we set
\[
f(t) := \exp(2\pi i\int_0^t p^\ast \alpha)\, .
\]
Then $f(0) = 1$ and 
\begin{align*}
f(t+1) &= \exp\left(2\pi i\int_0^{t+1} p^\ast \alpha\right) \\
& = \exp\left(2\pi i\int_0^t p^\ast \alpha\right) \exp\left(2\pi i\int_t^{t+1} p^\ast \alpha\right)\\
&= f(t) \exp\left(2\pi i\int_{S^1}\alpha \right)\, .
\end{align*}
We infer that $f$ is quasi-periodic with quasi-period $\exp(2\pi i\int_{S^1}\alpha)$.

Conversely, suppose that $f\: \Bbb R \to S^1$ is quasi-periodic, with $f(0) = 1$. Let $f^\flat\: \Bbb R \to \Bbb R$ be
the unique lift to the universal cover (i.e., $p\circ f^\flat = f$) satisfying $f^\flat(0) = 0$. Choose $s\in \Bbb R$ so that $\lambda = e^{2\pi i s}$.
Then 
\[
f^\flat(t+1) - f^\flat(t) \in s+ \Bbb Z\, ,
\]
for all $t$. By continuity, we infer that $f^\flat(t+1) - f^\flat(t)$ is a constant function, say $f^\flat(t+1) - f^\flat(t) = 2\pi c$. 

Set $g(t) = f^\flat(t) - 2\pi c t$. Then 
\[
g(t+1) = f^\flat(t+1) - 2\pi c (t+1) = f^\flat(t) + 2\pi c - 2\pi c (t+1)  = g(t)\, ,
\]
so $g(t)$ is $1$-periodic. Hence, there is a unique smooth function $\phi\: S^1 \to \Bbb R$ such that
\[
\phi(p(t)) = g(t) = f^\flat(t) - 2\pi c t\, .
\]
Equivalently, 
\[
f^\flat(t) = 2\pi c t + \phi(p(t)) =  \int_0^t p^*(cd\theta + d\phi)\, ,
\]
where $d\theta\in \Omega^1(S^1)$ is the arc length element (here we have used the identity $2\pi dt = p^\ast(d\theta)$).
Set $\alpha = cd\theta + d\phi$. Then
\[
f(t) = p(f^\flat(t)) = \exp(2\pi i \int_0^t p^*\alpha)\, .
\]
Hence, the two procedures outlined above are inverses. Both preserve smooth parameter families:
the inverse is characterized by $p^\ast\alpha=(2\pi i)^{-1}f^{-1}df$, and the forward map is defined by integration.
This establishes (i).

As for (ii), one first observes that the map 
\begin{align*}
Q(\Bbb R,S^1) &\to S^1 \times Q_\ast(\Bbb R,S^1)\\
f & \mapsto (f(0), t\mapsto f(t)f(0)^{-1})
\end{align*}
is a diffeomorphism. Then (ii) follows from (i).
\end{proof}

\begin{rem} The above proof shows that a quasi-periodic map $f\: \Bbb R\to S^1$ corresponds to an {\it exact} $1$-form $\alpha\in \Omega^1(S^1)$
if and only if $f$ is periodic and has winding number 0.
\end{rem}

\subsection{Quasi-periodic states}
For a finite dimensional complex Hilbert space $W$, we let $S(W)$ denote its unit sphere.
 Let $\Bbb P(W) = S(W)/U(1)$ be the associated projective space. Then the quotient map
 \[
 S(W) \to \Bbb P(W)
 \]
 is a $U(1)$-principal bundle with the {\it Fubini-Study connection}
(cf.~the proof of Lemma \ref{lem:aa-statement} below).
 
 \begin{defn}
A {\it normalized quasi-periodic state} is a quasi-periodic map $\psi\: \Bbb R\to S(W)$.
If $\lambda$ is the quasi-period of $\psi$ we occasionally write $\lambda = \lambda_\psi$.
We  sometimes denote the quasi-periodic state as a pair
$(\lambda,\psi)$.

Let
\[
\cal S(W) := Q(\Bbb R,S(W))
\]
be the diffeological space of normalized quasi-periodic states.
\end{defn}

As in \eqref{eqn:sharp},  a normalized quasi-periodic state $\psi$ determines  a smooth loop
\[
\psi^\sharp\: \Bbb R/\Bbb Z \to \Bbb P(W)\, .
\]

% begin{rem} If $(\lambda, \psi)$ is a normalized quasi-periodic state, then
 % $\psi$ is $\Bbb Z$-equivariant with respect to the translation action of $\Bbb Z$ on the real line and by
 % the action  on  $S(W)$ defined by $n\cdot x = \lambda^nx$. We let 
 % $S(W)_{\lambda}$ denote $S(W)$ with the latter action. 
% Then a quasi-periodic state with respect to a given $\lambda$ is the same thing as a smooth section of the sphere bundle
% \[
% S(W)_{\lambda} \times_{\Bbb Z} \Bbb R \to \Bbb R/\Bbb Z = S^1\, .
% \]
%\end{rem}

The group $U(1)$ acts on $\cal S(W)$ by $(\mu\cdot\psi)(t):=\mu\psi(t)$.
Let 
\[
\cal P(W) := \cal S(W)/U(1)
\]
be the quotient diffeological space.
Then one has a pullback square of diffeological spaces
\[
\xymatrix{
\cal S(W) \ar[r]\ar[d] & S(W) \ar[d] \\
\cal P(W) \ar[r] & \Bbb P(W)
}
\]
where the horizontal maps are defined by evaluation at $0$.
Indeed, a projective state and a unit vector over its value at $0$ determine a unique
normalized representative with that initial value, smoothly in families. It follows that
the left vertical map is a locally trivial diffeological principal $U(1)$-bundle.
Then $\cal P(W)$ is the (diffeological) space of {\it projective quasi-periodic states.}

There is a smooth map 
\[
\cal P(W) \to U(1) \times L\Bbb P(W)
\]
where $L\Bbb P(W)$ is the space of smooth free loops
in $\Bbb P(W)$. The map is given by $[\psi] \mapsto (\lambda_\psi, \psi^\sharp)$.

\begin{prop} \label{prop:projective-quasi-periodic}
There is a preferred diffeomorphism 
\[
\cal P(W) \cong L\Bbb P(W) \times \Omega^1(S^1) \, .
\]
\end{prop}

\begin{proof} We will define a bijection on the level of points and leave it to the reader
to verify that it defines a diffeomorphism of diffeological spaces.

The $U(1)$-principal bundle $S(W)\to\Bbb P(W)$ is the unitary frame bundle
of the canonical line bundle $\gamma_W\subset\Bbb P(W)\times W$.
The latter has a preferred connection $\nabla$, obtained by orthogonally
projecting the trivial connection onto $\gamma_W$.
The induced connection on $S(W)$ determines a smooth holonomy map
\[
L\Bbb P(W) \to U(1)\, ,
\]
which assigns to a smooth loop $g\: S^1 \to \Bbb P(W)$ an isometry of the fiber of $\gamma_W$ at $g(1)$, or equivalently,
an element of $U(1)$.

On the level of points,
we first define a smooth map 
\[
\cal S(W) @> (h_1,h_2) >>  L\Bbb P(W) \times \Omega^1(S^1)
\]
which factors through $\cal P(W)$. 
The map $h_1$ is given by $h_1(f) =  f^\sharp$, as in \eqref{eqn:sharp}.

We next  define $h_2$. Suppose that
 $f\: \Bbb R \to S(W)$ is quasi-periodic. Let $p\: \Bbb R \to U(1)$ be the universal cover:  $p(s) = \exp( 2\pi i s)$.
The connection $\nabla$
induces another lift of $f^\sharp$ (with initial lift $f(0)$),  to a quasi-periodic map $f_0\: \Bbb R \to S(W)$.
Therefore,  
\[
f(t)m(t) = f_0(t)\, ,
\] where $m\: \Bbb R \to U(1)$ is quasi-periodic and satisfies $m(0) = 1$.
Let $\alpha \in \Omega^1(S^1)$ correspond to $m$ by Lemma \ref{lem:qp-forms}.
Set 
\[
h_2(f)  :=  \alpha\, .
\]
Then we have defined the map $(h_1,h_2)\: \cal S(W) \to  L\Bbb P(W) \times \Omega^1(S^1)$. We leave it to the
reader to verify that this map
factors through $\cal P(W)$.

Conversely, suppose that $(g,\alpha)\in L\Bbb P(W) \times \Omega^1(S^1)$.
Choose an initial lift of $g(1)$ to $S(W)$.
Then $g$ lifts to a quasi-periodic map $g^\flat_0\: \Bbb R\to S(W)$ using the connection $\nabla$.
 Let $m\: \Bbb R \to U(1)$ be the quasi-periodic map that satisfies $m(0) = 1$ and which corresponds to $\alpha$ as in Lemma \ref{lem:qp-forms}. 
 Set $g^\flat(t) := m(t)^{-1} g^{\flat}_0(t)$. If we had changed the lift of $g(1)$, then $g^\flat$ changes by some factor $\rho \in U(1)$,
 so $[g^\flat] \in \cal P(W)$ is well-defined.
The map $(g,\alpha) \mapsto [g^\flat]$ is smooth and defines the inverse to the map $(h_1,h_2)$.
 \end{proof}

\begin{rem} Let $(\lambda,f)\in\cal S(W)$ be a normalized quasi-periodic state.
The equation $f(t)=f_0(t)m(t)^{-1}$ evaluated at $t=1$ gives
\[
f(1)=f_0(1)\exp\left(-2\pi i\int_{S^1}\alpha\right).
\]
Since $f(1)=f(0)\lambda=f_0(0)\lambda$ and
$f_0(1)=f_0(0)\mu_{f^\sharp}$, where $\mu_{f^\sharp}$ is the
holonomy of the loop $f^\sharp\in L\Bbb P(W)$, we obtain
\begin{equation}\label{eqn:boxed}
\lambda=\mu_{f^\sharp}\exp\left(-2\pi i\int_{S^1}\alpha\right).
\end{equation}
\end{rem}

\begin{cor} \label{cor:projective-quasi-periodic}
There is a $U(1)$-equivariant diffeomorphism 
\[
\cal S(W) \cong L_{S(W)}\Bbb P(W)  \times \Omega^1(S^1)\, ,
\]
where $ L_{S(W)}\Bbb P(W)$ is the diffeological space given by the
fiber product of the diagram \[ L\Bbb P(W) \to \Bbb P(W)\leftarrow S(W) \, .\]
\end{cor}

\begin{proof}
The diffeomorphism is given by $(\lambda,f) \mapsto ((f^\sharp,f(0)),\alpha)$, where the notation is as above.
\end{proof}

\begin{rem} When $W = \Bbb C$, Corollary \ref{cor:projective-quasi-periodic}
coincides with the statement of Lemma \ref{lem:qp-forms}(ii), since $L_{S(\Bbb C)}(\Bbb P(\Bbb C)) = S^1$.
\end{rem}

\section{Floquet states} \label{sec:Floquet}
We let
\[
 H(W)
\] 
denote the space of Hermitian operators $W\to W$.  A {\it periodic Hamiltonian} is a smooth $1$-periodic map $h\: \Bbb R \to H(W)$.
Equivalently it is a smooth free loop $\Bbb R/\Bbb Z \to H(W)$.
The space of periodic Hamiltonians is therefore the diffeological space $LH(W)$. 

\begin{defn} For a periodic Hamiltonian $h$,
a {\it Floquet state} is a normalized quasi-periodic state $(\lambda,\psi)$ that satisfies the Schr\"odinger equation
\[
i\frac{d\psi}{dt} = h(t) \psi(t)\, .
\]
\end{defn}

\begin{rem} \label{rem:evolution} Let $U(W)$ denote the Lie group of linear isometries of $W$, i.e., unitary operators $W\to W$. Given a periodic Hamiltonian $h$,
one has an associated {\it evolution operator} $\cal U := \cal U\! (h)\: \Bbb R\times \Bbb R\to U(W)$ which
characterizes the time evolution of solutions to the Schrödinger equation:
\[
\psi(t) = \cal U\! (t,t_0)\psi(t_0)\, .
\]
Equivalently, the operator $\cal U\! (t,t_0)$  solves the initial value problem
\[
i\cal U'\! (t,t_0) = h(t) \cal U\! (t,t_0), \qquad \cal U\! (t_0,t_0) = I\, .
\]
The spectrum of  the operator $\cal U\!(t+1,t)$ does not depend on $t$.
Moreover, a Floquet state $(\lambda,\psi)$  solves  the equation
\[
 \cal U\! (t+1,t) \psi(t) = \lambda\psi(t)\, .
\]
\end{rem}

\begin{defn} \label{defn:good-periodic-hamiltonian}
We say that a periodic Hamiltonian $h\: \Bbb R \to H(W)$ is {\it good} if the spectrum of $\cal U\! (1,0)$ is non-degenerate, i.e.,
the spectrum $\Lambda\subset S^1$ of  $\cal U\!(1,0)$ has $n$ distinct eigenvalues. We call $\Lambda$ the set of {\it quasi-periods}.
Let 
\[
\cal H_F(W) \subset LH(W)
\] 
denote the  space of good periodic Hamiltonians with the induced diffeology.
\end{defn}

The monodromy map is smooth, and the simple-spectrum locus in $U(W)$ is open.
Thus $\cal H_F(W)$ is an open subset of the nuclear Fr\'echet space
$LH(W)=C^\infty(S^1,H(W))$. In particular, it is Hausdorff and metrizable,
hence paracompact, and is therefore nice.

\section{The model of \texorpdfstring{$\cal H_F(W)$}{H\_F(W)}} \label{sec:H_F-model}

Let $V(W)$ be the Stiefel manifold of orthonormal $n$-frames in $W$. Then the unitary group $U(n)$ acts freely on the right
of $V(W)$. Let $T = U(1)^{\times n} \subset U(n)$ be the maximal torus. Then $F(W) = V(W)/T$ is the complete flag manifold, and
the map $V(W) \to F(W)$ is a $T$-principal bundle. A point of $F(W)$ can be described as an ordered $n$-tuple of lines
$(L_1,\dots,L_n)$ which are pairwise orthogonal.
Orthogonally projecting the trivial connection on $F(W)\times W$ onto
the $n$ tautological line bundles induces a canonical $T$-connection on
$V(W)\to F(W)$.

Let $C_n(S^1)$ be the configuration space of finite subsets of $S^1$ having cardinality $n$.
Let
\[
\tilde C_n(S^1) \to C_n(S^1)
\]
be the universal cover.  The total space $\tilde C_n(S^1)$ is contractible
and the base has fundamental group $\Bbb Z$. In particular, one has a diffeomorphism
\[
\tilde C_n(S^1)/\Bbb Z \cong  C_n(S^1)\, .
\]
The following concrete model for $\tilde C_n(S^1)$ is used throughout the paper.

\begin{defn} Let $\tilde C_n(S^1)\subset \Bbb R^n$
be the subspace given by the set of points $(t_1,\dots t_n)\in \Bbb R^n$ that satisfy the inequalities
\[
t_1 < t_2 < \cdots < t_n < t_1+1\, .
\]
The covering projection $\tilde C_n(S^1) \to C_n(S^1)$ is defined by
\[
(t_1,\dots,t_n) \mapsto \{\exp(2\pi i t_1),\dots,\exp(2\pi i t_n)\} \, .
\] 
The action by the deck generator is $(t_1,\dots,t_n)\mapsto (t_2,\dots,t_n,t_1+1)$.
\end{defn}

Our model of $\cal H_F(W)$ is constructed from a pair of cartesian squares
\begin{equation}\label{eqn:pair-cartesian}
\xymatrix{
V(W) \ar[r]^{/T}   & F(W) \\
\tilde{\cal H}_V(W) \ar[r]^{/T}  \ar[u] &  \tilde{\cal H}_{F}(W) \ar[u] \ar[r]^(.45){/\Bbb Z} \ar[d] &\cal H_{F}(W)  \ar[d] \\
& \tilde C_n(S^1) \ar[r]_(.45){/\Bbb Z} & C_n(S^1)
}
\end{equation}
in which each horizontal map is identified with the quotient of a free action, either by the $n$-torus $T$ or by the integers $\Bbb Z$.
The horizontal maps in the middle of the diagram are locally trivial diffeological principal bundles, since they are base changes along
smooth maps of locally trivial smooth principal bundles of smooth manifolds. Then the composition of the middle arrows of the diagram
is a locally trivial diffeological principal $G$-bundle
\[
\tilde{\cal H}_V(W) \to \cal H_{F}(W) \, .
\]
Explicitly, a point of $\tilde{\cal H}_{F}(W)$ is a pair
\[
(h,\bar t)
\]
in which $h\in \cal H_F(W)$ and $\bar t = (t_1,\dots,t_n)$ is such that the underlying unordered set of the $n$-tuple 
$\lambda := (\exp(2\pi i t_1),\dots,\exp(2\pi i t_n))$ is the set of quasi-periods of $h$.
A point of $\tilde{\cal H}_V(W)$ consists of a triple
\[
(h,\bar t, \bold e{\, \smallbullet})\, ,
\]
in which $(h,\bar t)\in\tilde{\cal H}_{F}(W)$ and
$\bold e{\,\smallbullet}=(\bold e_1,\dots,\bold e_n)\in V(W)$ is an orthonormal frame satisfying
\[
\mathcal U_h(1,0)\bold e_j=\exp(2\pi i t_j)\bold e_j,
\qquad j=1,\dots,n,
\]
where $\mathcal U_h$ is the evolution operator associated with $h$.

We equip $\tilde{\cal H}_{F}(W) \subset \cal H_F(W) \times \tilde C_n(S^1)$ and
$\tilde{\cal H}_V(W) \subset \tilde{\cal H}_{F}(W) \times V(W)$ with the subdiffeology.
Let 
\[
G = T\rtimes \Bbb Z
\] be the semi-direct product of $\Bbb Z$ acting on $T$ 
in which the generator $1\in \Bbb Z$ acts by the cyclic permutation
\[
(u_1,\dots,u_n)\mapsto(u_n,u_1,\dots,u_{n-1})\, .
\]
Then $G$ acts on $\tilde {\cal H}_V(W)$. The  quotient by $T \subset G$ and then $H =G/T$ yields the
horizontal maps in the middle of diagram \eqref{eqn:pair-cartesian}.
Then applying Proposition \ref{prop:composite-bundle}, we arrive at the following statement.

\begin{prop} The map $\tilde{\cal H}_V(W) \to \cal H_F(W)$ is a locally trivial diffeological principal $G$-bundle.
In particular, there 
are diffeomorphisms
\[
\tilde{\cal H}_V(W)/G \cong \tilde{\cal H}_F(W)/\Bbb Z \cong \cal H_F(W)\, .
\]
\end{prop}

We now generalize the notion of quasi-periodicity to actions of $T$. Let $E \to M$ be a smooth principal $T$-bundle.
Let $E_n(S^1) \subset (S^1)^n=T$ be the configuration space given by the set of ordered $n$-tuples
$(\lambda_1,\dots,\lambda_n) \in T$ such that
$\lambda_i \ne \lambda_j$ for $i\ne j$.  

\begin{defn} A smooth map $f\: \Bbb R \to E$ is {\it quasi-periodic} (with respect to the action of  $T$) if there is a $\lambda \in T$ 
such that $f(t+1) = f(t)\lambda$ for all $t\in \Bbb R$. In addition, we say that $f$ is {\it good} if $\lambda \in E_n(S^1)$.
\end{defn}

\begin{defn} \label{defn:q}
Let $\tilde{\cal V}(W)$ be the diffeological space whose points are pairs
\[
(\gamma,\bar t)\, ,
\]
in which $\gamma=(\gamma_1,\dots,\gamma_n)\:\Bbb R\to V(W)$ is good quasi-periodic and
$\bar t=(t_1,\dots,t_n)\in\tilde C_n(S^1)$ satisfies
\[
\gamma_j(s+1)=\exp(2\pi i t_j)\gamma_j(s),
\qquad s\in\Bbb R,\quad j=1,\dots,n.
\]
We equip this space with the subdiffeology inherited from
\[
C^\infty(\Bbb R,V(W))\times\tilde C_n(S^1)\, ,
\]
where the mapping space has the functional diffeology.
\end{defn}

The action by $G $ on $\tilde{\cal V}(W)$ is as follows: 
Let $(\gamma,\bar t) \in \tilde {\cal V}(W)$ be a point.
If $u\in T$, then we set $(\gamma,\bar t)\cdot u =  (\gamma\cdot u, \bar t)$, where
the action is pointwise in the first factor.
The action by the generator of $\Bbb Z$ is induced by
cyclically permuting the factors of $V(W)$ and acts by the deck transformation on the second factor:
i.e., if $\gamma = (\gamma_1,\dots,\gamma_n)$ and $\bar t = (t_1,\dots,t_n)$,
then $(\gamma,\bar t) \mapsto ((\gamma_2,\dots, \gamma_n,\gamma_1), (t_2,\dots,t_n,t_1+ 1))$.

\begin{prop} \label{prop:v-h} There is an equivariant diffeomorphism
\[
\tilde {\cal V}(W) \cong \tilde{\cal H}_V(W)\, .
\]
\end{prop}

\begin{proof} Let $(\gamma,\bar t)\in \tilde{\cal V}(W)$ be given. Then $\gamma = (\gamma_1,\dots,\gamma_n)$, where
$\gamma_i\: \Bbb R \to S(W)$ is quasi-periodic.
One has an induced map $\gamma^\sharp\: \Bbb R/\Bbb Z \to F(W)$.
Using the canonical connection, let $\gamma^0\:\Bbb R\to V(W)$ be the
quasi-periodic horizontal lift with $\gamma^0(0)=\gamma(0)$. Then
\[
\gamma(t) = \gamma^0(t)f(t)\, ,
\]
where $f(t) \: \Bbb R\to T$ is quasi-periodic and satisfies $f(0) = (1,\dots,1)$. Let $f_i$ be the $i$th component of $f$.
Then $f = (f_1,\dots,f_n)$, $f_i\: \Bbb R \to U(1)$ is quasi-periodic,
and  $f_i(0) = 1$. Set $\bold e_i(t) := \gamma_i^0(t)$. Then $(\bold e_1(t),\dots,\bold e_n(t))$
is a moving basis for $W$.

By the Leibniz rule, we have
\begin{align*}
\tfrac{d}{dt}(f_j(t){\bold e}_j(t))&=  \dot{f_j}(t){\bold e}_j(t) + f_j(t){\dot {\bold  e}}_j(t) \, ,\\
           &= (\dot{f}_j(t) + if_j(t)\alpha_{jj}(t)) {\bold e}_j(t) + i\sum_{k \in I \setminus \{j\} } f_j(t)\alpha_{jk}(t) {\bold e}_k(t) \, ,
\end{align*}
 for $j\in I:=\{1,\ldots,n\}$,
where ${\dot{\bold  e}}_j(t) := i\sum_{k=1}^n \alpha_{jk}(t) {\bold e}_k(t)$. 
Then
\begin{align*}
f_j(t)h({\bold e}_j(t)) &= h(f_j(t){\bold e}_j(t))  \, ,\\
& =i \tfrac{d}{dt}(f_j(t){\bold e}_j(t)) \, ,\\
&=  (i\dot{f}_j(t)- f_j(t) \alpha_{jj}(t)) {\bold e}_j(t) - \sum_{k\in I \setminus \{j\}} f_j(t)\alpha_{jk}(t) {\bold e}_k(t)\, .
\end{align*}
Dividing by $f_j(t)$, we obtain the expression
\[
h({\bold e}_j(t)) =  (i\tfrac{\dot{f}_j(t)}{f_j(t)}-  \alpha_{jj}(t)) {\bold e}_j(t) - \sum_{k \in I \setminus \{j\} } \alpha_{jk}(t) {\bold e}_k(t)\, .
\]
Set $\rho_j(t) := i\frac{\dot{f}_j(t)}{f_j(t)}- \alpha_{jj}(t)$. 
Hence, the matrix entries of $h$ in the moving basis, with $h(\bold e_k)=\sum_j h_{jk}\bold e_j$, are
\begin{equation} \label{eqn:hjk}
h_{jk}(t) = \begin{cases}
\rho_j(t), & j= k, \\
-\alpha_{kj}(t), & j \ne k\, ,
\end{cases}
\end{equation}
Invariantly, $h(t)=i\dot\gamma(t)\gamma(t)^{-1}$, where $\gamma(t)$ is regarded as an isometry $\Bbb C^n\to W$.
Thus $h$ is Hermitian and $1$-periodic, and its evolution satisfies $\mathcal U_h(t,0)\gamma(0)=\gamma(t)$.
Then we define a map $\Phi\: \tilde{\cal V}(W)  \to \tilde{\cal H}_V(W)$ by
\[
\Phi(\gamma,\bar t) :=  (h,\bar t, \bold e{\, \smallbullet})\, ,
\]
where $h$ is as in \eqref{eqn:hjk},  and  $\bold e{\, \smallbullet}$ is the orthonormal frame
$(\bold e_1(0),\dots,\bold e_n(0))$.  It is  clear that $\Phi$ is $G$-equivariant. 
We will construct its inverse.

Let $(h,\bar t,e_\bullet)\in\tilde{\mathcal H}_V(W)$, with $\bar t=(t_1,\ldots,t_n)$.
For each $j$, let $\gamma_j(t)$ solve
\[
i\dot\gamma_j(t)=h(t)\gamma_j(t),\qquad \gamma_j(0)=e_j.
\]
Then $\gamma(t)=(\gamma_1(t),\ldots,\gamma_n(t))$ is an orthonormal frame for every $t$.
Since $e_j$ is a monodromy eigenvector with eigenvalue $\lambda_j=e^{2\pi i t_j}$, periodicity of $h$ gives
\[
\gamma_j(t+1)=\lambda_j\gamma_j(t).
\]
Thus $\gamma$ is a good quasi-periodic frame, and the inverse is
\[
\Psi\:\widetilde{\mathcal H}_V(W)\longrightarrow\widetilde{\mathcal V}(W),
\qquad (h,\bar t,e_\bullet)\longmapsto(\gamma,\bar t).
\]
Uniqueness of solutions to the Schrödinger equation implies that
$\Phi$ and $\Psi$ are inverse to one another. The constructions are
smooth and $G$-equivariant. Therefore $\Phi$ is a $G$-equivariant
diffeomorphism.
\end{proof}

\section{The diffeomorphism type of \texorpdfstring{$\cal H_F(W)$}{H\_F(W)}} \label{sec:diffeo-type-H}
The purpose of this section is to identify the diffeomorphism type
of $\cal H_F(W)$. Recall by Proposition \ref{prop:v-h} that there is a diffeomorphism
\[
\tilde{\cal V}(W)/G  \cong \cal H_F(W)\, .
\]
The following result identifies the diffeomorphism type of $\tilde{\cal V}(W)$. 

\begin{prop} \label{prop:diffeo-tilde-Q} There is a canonical $G $-equivariant diffeomorphism  
\[
\tilde{\cal V}(W) \cong \tilde C_n(S^1) \times LV(W)  \, .
\]
\end{prop}

\begin{rem} The action of $G$ on $\tilde C_n(S^1)\times LV(W)$ is as follows.
The torus $T$ acts trivially on $\tilde C_n(S^1)$ and pointwise on the right of $LV(W)$.
Write $\delta\bar t=(t_2,\ldots,t_n,t_1+1)$ for the deck generator.
On the product, the generator of $\Bbb Z$ acts by
\[
(\bar t,\gamma)\longmapsto
\left(\delta\bar t,\ s\longmapsto
\bigl(\gamma_2(s),\ldots,\gamma_n(s),e^{-2\pi is}\gamma_1(s)\bigr)\right).
\]
The last phase factor compensates for the increase of the last lifted quasi-period by one.
\end{rem}

\begin{proof}[Proof of Proposition \ref{prop:diffeo-tilde-Q}] Identify the Stiefel manifold $V(W)$  with 
the space of linear isometries $\Bbb C^n\to W$.
% and the unitary group $U(n)$ is identified with the space of linear isometries of $\Bbb C^n$. Let
% $E\: \Bbb C^n \to W$ be a fixed isometry corresponding to  an orthonormal basis of $W$.

% Then one has a  diffeomorphism
% \begin{align}\label{eqn:v-u}
% \rho\: U(n) &@> \cong >> V(W) \\
 % A & \mapsto E\circ  A\, . \nonumber
% \end{align}
Then a smooth map $\Psi\: \tilde C_n(S^1) \times LV(W)  \to \tilde{\cal V}(W)$ is defined 
by
\[
\Psi(\bar t,\gamma) = \bigl(s\mapsto\gamma(s)\circ D_{\bar t}(s),\bar t\bigr)\, ,
\]
where
$\bar t = (t_1,\dots,t_n) \in \tilde C_n(S^1)$ and
$D_{\bar t}(s)\: \Bbb C^n \to \Bbb C^n$ is the linear transformation corresponding to the diagonal matrix
\[
D_{\bar t}(s)=\operatorname{diag}\bigl(e^{2\pi i t_1s},\dots,e^{2\pi i t_ns}\bigr)\, .
\]
The inverse of $\Psi$ is given by
\[
\Psi^{-1}(\psi,\bar t) := \bigl(\bar t,s\mapsto\psi(s)\circ D_{-\bar t}(s)\bigr)\, . 
\]
It follows that $\Psi$ is a diffeomorphism.
We leave the verification of equivariance to the reader.
\end{proof}

\begin{cor} \label{cor:diffeo-tilde-Q} There is a canonical diffeomorphism
\[
\tilde{\cal V}(W) \cong \tilde C_n(S^1) \times \Omega U(W) \times V(W)\, .
\]
\end{cor}

\begin{proof} There is a canonical diffeomorphism
\[
LV(W)\longrightarrow\Omega U(W)\times V(W),
\qquad
\gamma\longmapsto\bigl(s\longmapsto\gamma(s)\gamma(0)^{-1},\gamma(0)\bigr).
\]
Here the composition $\gamma(s)\gamma(0)^{-1}$ is an operator on $W$.
The inverse sends $(K,E)$ to the loop $s\mapsto K(s)E$.
The result follows from Proposition \ref{prop:diffeo-tilde-Q}.
\end{proof}

We now identify the action of $G=T\rtimes\Bbb Z$ on
the right side of Corollary \ref{cor:diffeo-tilde-Q} making the diffeomorphism equivariant.
The subgroup $T$ acts trivially on $\tilde C_n(S^1)$ and $\Omega U(W)$,
and by the usual right action on $V(W)$.
For $E=(e_1,\ldots,e_n)$, write $rE=(e_2,\ldots,e_n,e_1)$ and $L_1=\Bbb C e_1$.
If $P_{L_1}$ denotes orthogonal projection onto $L_1$, the generator of $\Bbb Z$ acts by
\[
\delta(\bar t,K,E)
=\bigl(\delta\bar t,\ s\longmapsto K(s)e^{-2\pi isP_{L_1}},\ rE\bigr).
\]
The phase factor depends only on $L_1$, so this action is compatible with the torus action.
We leave it to the reader to verify the following.

\begin{add} \label{add:diffeo-tilde-Q} The above recipe provides a smooth action of $G$ on
$\tilde C_n(S^1)\times\Omega U(W)\times V(W)$ which makes
the diffeomorphism of Corollary \ref{cor:diffeo-tilde-Q} equivariant.
\end{add}

The next step is to determine the quotient of the right side of Corollary \ref{cor:diffeo-tilde-Q} by the action of $G $.
By the double quotient formula  \eqref{eqn:double-quotient-formula}
\[
\tilde {\cal V}(W)/G  \cong (\tilde {\cal V}(W)/T)/\Bbb Z \, .
\] 
By Corollary \ref{cor:diffeo-tilde-Q} and Addendum \ref{add:diffeo-tilde-Q}, there is a $\Bbb Z$-equivariant diffeomorphism
\[
\tilde{\cal V}(W)/T\cong
Y:=\tilde C_n(S^1)\times\Omega U(W)\times F(W),
\]
where $F(W)=V(W)/T$. Write
\[
\mathcal L=(L_1,\ldots,L_n),\qquad
r\mathcal L=(L_2,\ldots,L_n,L_1).
\]
The deck action is
\[
\delta(\bar t,K,\mathcal L)
=\bigl(\delta\bar t,\ s\longmapsto K(s)e^{-2\pi isP_{L_1}},\ r\mathcal L\bigr).
\]
Let
\[
w(K):=\operatorname{wind}(\det K)
=\frac{1}{2\pi i}\int_0^1\operatorname{tr}\bigl(K(s)^{-1}K'(s)\bigr)\,ds\in\Bbb Z.
\]
Since $\det(e^{-2\pi isP_{L_1}})=e^{-2\pi is}$, application of $\delta$ decreases $w(K)$ by one.
Thus every orbit has exactly one representative in
\[
Y_0:=\tilde C_n(S^1)\times\Omega_0U(W)\times F(W),
\]
where $\Omega_0U(W)=\{K\in\Omega U(W)\mid w(K)=0\}$.
More precisely, there is a $\Bbb Z$-equivariant diffeomorphism
\[
Y\longrightarrow Y_0\times\Bbb Z,
\qquad z=(\bar t,K,\mathcal L)\longmapsto
\bigl(\delta^{w(K)}z,-w(K)\bigr),
\]
with inverse $(z_0,m)\mapsto\delta^m z_0$.
Here $\Bbb Z$ acts trivially on $Y_0$ and by translation on the last factor.
Smoothness follows because $w$ is locally constant and the deck action is smooth.

For $K\in\Omega_0U(W)$, the determinant has a unique based lift
$f\in\Omega\Bbb R$ satisfying
\[
\det K(s)=e^{2\pi if(s)},\qquad f(0)=0.
\]
It is given by
\[
f(s)=\frac{1}{2\pi i}\int_0^s\operatorname{tr}\bigl(K(u)^{-1}K'(u)\bigr)\,du.
\]
Scalar multiplication gives a canonical diffeomorphism
\[
\Omega_0U(W)\longrightarrow\Omega SU(W)\times\Omega\Bbb R,
\qquad K\longmapsto\bigl(e^{-2\pi if/n}K,f\bigr),
\]
whose inverse is $(A,f)\mapsto e^{2\pi if/n}A$.
This construction requires no choice of basis for $W$.
Consequently, there is a $\Bbb Z$-equivariant diffeomorphism
\begin{equation} \label{eqn:Z-eq}
\tilde{\cal V}(W)/T\cong
\tilde C_n(S^1)\times\Omega SU(W)\times\Omega\Bbb R\times F(W)\times\Bbb Z,
\end{equation}
where $\Bbb Z$ acts by translation on the last factor and trivially on all other factors.

To complete the identification of $\cal H_F(W)$, we take the quotient by $\Bbb Z$ in \eqref{eqn:Z-eq}.
This is a special case of the following fact.

\begin{lem}[Untwisting the action] \label{lem:untwist} Let $X$ be a diffeological space equipped with the smooth right action of 
a discrete group $H$. Let $H$ act on itself by right translation. Give the product $X\times H$ the diagonal action. 
Let $X^t$ denote $X$ with the trivial $H$-action. Then there is
an $H$-equivariant diffeomorphism
\[
X \times H \cong X^t\times H\, .
\]
In particular, the quotient $(X\times H)/H$ is diffeomorphic to $X$.
\end{lem}

\begin{proof} 
Let $\phi\: X\times H \to X^t \times H$ be the diffeomorphism
defined by $(x,h) \mapsto (xh^{-1},h)$. Then $\phi$ is $H$-equivariant.
\end{proof}

The preceding identifications yield the following result.

\begin{thm} \label{thm:diffeo-hamiltonians}
There is a preferred diffeomorphism, natural under unitary isomorphisms of $W$,
\[
\cal H_F(W)\cong\tilde C_n(S^1)\times\Omega SU(W)\times\Omega\Bbb R\times F(W).
\]
\end{thm}

\begin{proof}
By the double quotient formula \eqref{eqn:double-quotient-formula},
\[
\cal H_F(W)\cong\tilde{\cal V}(W)/G
\cong(\tilde{\cal V}(W)/T)/\Bbb Z.
\]
Under the equivariant diffeomorphism \eqref{eqn:Z-eq}, the group $\Bbb Z$ acts by translation on the last factor and trivially on all other factors. Taking the quotient therefore removes the last factor and gives the asserted product, as in Lemma~\ref{lem:untwist}. All maps used above are canonically specified and commute with unitary isomorphisms of $W$, so the resulting diffeomorphism is preferred and natural.
\end{proof}

\begin{construct}[Explicit description of the diffeomorphism]
\label{prop:intrinsic-product-decomposition}
With the period fixed to be one and the time origin zero, we describe explicitly the diffeomorphism of Theorem~\ref{thm:diffeo-hamiltonians} and its inverse.

For $h\in\mathcal H_F(W)$, write
\[
\mathcal U_h(t):=\mathcal U_h(t,0),
\qquad
M_h:=\mathcal U_h(1),
\qquad
\sigma_h:=-\frac{1}{2\pi}
                  \int_0^1\operatorname{tr}h(u)\,du.
\]
The evolution equation gives
\[
\det\mathcal U_h(t)
=
\exp\left(-i\int_0^t\operatorname{tr}h(u)\,du\right),
\]
so $\det M_h=e^{2\pi i\sigma_h}$.

There is a unique tuple
\[
\bar x_h=(x_1,\ldots,x_n)\in\widetilde{C}_n(S^1)
\]
such that
\[
\{e^{2\pi i x_1},\ldots,e^{2\pi i x_n}\}
   =\operatorname{spec}(M_h),
\qquad
\sum_{j=1}^n x_j=\sigma_h.
\]
Indeed, for any lift of the spectrum, the difference
$\sum_jx_j-\sigma_h$ is an integer. The deck generator
\[
(x_1,\ldots,x_n)\longmapsto
(x_2,\ldots,x_n,x_1+1)
\]
increases the sum by one. Thus the trace condition selects
exactly one lift.

Let $L_j$ be the eigenspace of $M_h$ with eigenvalue
$e^{2\pi i x_j}$, and let $P_{L_j}$ be its orthogonal
projection. Define
\[
\mathcal L_h:=(L_1,\ldots,L_n)\in F(W),
\qquad
B_h:=\sum_{j=1}^n x_jP_{L_j}.
\]
Then $B_h$ is Hermitian,
\[
M_h=e^{2\pi iB_h},
\qquad
\operatorname{tr}B_h=\sigma_h.
\]
Consequently,
\[
K_h(t):=\mathcal U_h(t)e^{-2\pi itB_h}
\]
is a smooth based loop in $U(W)$. Its periodicity follows
from $\mathcal U_h(t+1)=\mathcal U_h(t)M_h$.

Define
\[
f_h(t):=
-\frac{1}{2\pi}\int_0^t\operatorname{tr}h(u)\,du
-t\sigma_h.
\]
This function satisfies
\[
f_h(0)=0,\qquad f_h(t+1)=f_h(t),
\qquad
\det K_h(t)=e^{2\pi if_h(t)}.
\]
Thus $f_h\in\Omega\mathbb R$, and the determinant of $K_h$
has winding number zero. Set
\[
A_h(t):=e^{-2\pi if_h(t)/n}K_h(t).
\]
Then $A_h\in\Omega SU(W)$. The diffeomorphism is
\[
\Phi_W(h):=
\bigl(\bar x_h,A_h,f_h,\mathcal L_h\bigr).
\]
To see how this formula arises from the preceding argument, start with any spectral lift
$\bar t=(t_1,\ldots,t_n)$ and its associated based loop $K$. The determinant identity above gives
\[
w(K)=\sigma_h-\sum_{j=1}^n t_j.
\]
Applying $\delta^{w(K)}$ changes the sum of the lifted quasi-periods to $\sigma_h$.
Thus the zero-winding representative is precisely $(\bar x_h,K_h,\mathcal L_h)$.
The scalar determinant splitting then gives the factors $A_h$ and $f_h$.

To construct the inverse of $\Phi_W$, take
\[
\bigl(\bar x,A,f,\mathcal L\bigr)
\in
\widetilde{C}_n(S^1)\times
\Omega SU(W)\times
\Omega\mathbb R\times F(W),
\]
where $\mathcal L=(L_1,\ldots,L_n)$, and put
\[
B:=\sum_{j=1}^n x_jP_{L_j},
\qquad
K(t):=e^{2\pi if(t)/n}A(t),
\qquad
\mathcal U(t):=K(t)e^{2\pi itB}.
\]
Define
\[
h(t):=i\mathcal U'(t)\mathcal U(t)^{-1}
     =iK'(t)K(t)^{-1}-2\pi K(t)BK(t)^{-1}.
\]
This is a smooth, $1$-periodic Hermitian family.
Its evolution operator based at zero is $\mathcal U$,
and its monodromy is $e^{2\pi iB}$, which has simple
spectrum. Moreover,
\[
\operatorname{tr}h(t)
=-2\pi f'(t)-2\pi\sum_{j=1}^n x_j,
\]
and hence
\[
-\frac{1}{2\pi}\int_0^1\operatorname{tr}h(t)\,dt
=\sum_{j=1}^n x_j.
\]
The normalization above therefore recovers $\bar x$,
and then recovers $\mathcal L$, $K$, $f$, and $A$.
Conversely, reconstructing $h$ from $\Phi_W(h)$ returns
the original Hamiltonian.

Both constructions are smooth in the diffeological sense.
Indeed, solutions of the evolution equation depend
smoothly on finite-dimensional smooth parameters.
On the simple-spectrum locus, eigenvalues and spectral
projections admit smooth local choices; the trace
normalization selects a locally constant deck translate
of each local spectral lift. The remaining operations
are integration, differentiation, multiplication, and
exponentiation of smooth families. These observations
also establish smoothness of the displayed inverse.

Finally, if $Q:W\to W'$ is unitary, replacing $h$ by
$QhQ^{-1}$ leaves $\bar x_h$ and $f_h$ unchanged,
replaces $A_h$ by $QA_hQ^{-1}$, and replaces
$\mathcal L_h$ by $(QL_1,\ldots,QL_n)$.
Thus $\Phi_W$ is natural under unitary isomorphisms.
\end{construct}

\begin{rem}
The determinant is intrinsically defined by the action on the one-dimensional space $\bigwedge^nW$, and the factors $e^{\pm2\pi if/n}$ act by scalar multiplication. Neither operation requires a basis or a choice of a nonzero vector in $\bigwedge^nW$. An orthonormal basis is needed only to identify $SU(W)$ and $F(W)$ with $SU(n)$ and $F(\mathbb C^n)$.
\end{rem}

\subsection{The cohomology of $\cal H_F(W)$}
The spaces $\tilde C_n(S^1)$ and $\Omega \Bbb R$  are contractible. 
Consequently,

\begin{cor} There is a homotopy equivalence of spaces
\[
\cal H_F(W) \simeq \Omega SU(W) \times F(W)\, .
\]
\end{cor}

\begin{cor} The rational cohomology ring of $\cal H_F(W)$ is isomorphic to the graded tensor product
\[
\Bbb Q[s_1,s_2,\dots s_{n-1}] \otimes \Bbb Q[t_1,t_2,\dots t_{n}]/\Sigma\, ,
\]
in which $\deg s_j = 2j, \deg t_j = 2$ and $\Sigma \subset \Bbb Q[t_1,t_2,\dots t_{n}]$ is the ideal
generated by the elementary symmetric polynomials $e_1,\dots,e_n$ in $t_1,\dots,t_n$.
\end{cor}

\begin{proof} The space $\Omega SU(W)$ is equipped with the $\cal D$-topology.
Let
$\ubar \Omega SU(W)$ denote the space of continuous based loops equipped with the compact-open topology. Then  the inclusion
\[
\Omega SU(W)  \to \ubar \Omega SU(W)
\]
is a weak homotopy equivalence \cite{Kihara}. Hence these two loop spaces have the same rational cohomology.
As is well-known \cite{Bott1958}, it is a polynomial algebra with generators $s_j$ of degree $2j$. The other tensor factor
in the statement of the corollary is the rational cohomology of $F(W)$ \cite{Borel1953}.
\end{proof}

\section{The Adiabatic Floquet Conjecture} \label{sec:Floquet-conjecture}
Suppose that $h\: [t_0,t_1] \to H(W)$ is a smooth family of Hamiltonians, not necessarily periodic. Let $\psi\: [t_0,t_1] \to S(W)$ be the solution to the Schrödinger equation
with normalized initial condition $\psi_0 = \psi(t_0) \in S(W)$. Let
\[
[\psi] \: [t_0,t_1] \to \Bbb P(W)
\]
denote the projection of $\psi$.  Note $([\psi],\psi(t_0))$ defines a lifting problem for 
the principal $U(1)$-bundle
\[
p\: S(W) @>>> \Bbb P(W)\, ,
\]
where the latter is equipped with the canonical Fubini-Study connection. The connection
in turn defines a ``geometric'' lift
\[
\Gamma([\psi],\psi(t_0))\:[t_0,t_1]\longrightarrow S(W).
\]

For a ray $\ell\in\Bbb P(W)$ and a Hermitian operator $a\in H(W)$, write
\[
\langle\ell|a|\ell\rangle:=\langle av,v\rangle,
\]
where $v\in\ell$ is any unit vector. This is independent of the choice of $v$.

\begin{lem}[Aharonov--Anandan] \label{lem:aa-statement} One has
\[
\psi(t) = \exp\left( -i \int_{t_0}^t \langle [\psi](u)|h(u)|[\psi](u)\rangle du\right) \, \Gamma([\psi],\psi(t_0))(t)\, .
\] 
\end{lem}

\begin{rem} Lemma \ref{lem:aa-statement} states that the lifts $\psi$ and  $\Gamma([\psi],\psi(t_0))$ coincide up to the dynamical 
factor $\exp\left( -i \int_{t_0}^t \langle [\psi](u)|h(u)|[\psi](u)\rangle du\right)$.
\end{rem}

\begin{proof}[Proof of Lemma  \ref{lem:aa-statement}] Set $\phi(t) = \Gamma([\psi],\psi(t_0))(t)$. We need to prove that
\[
\psi(t) = e^{-i\theta(t)}\phi(t), \qquad \theta(t) = \int_{t_0}^t \langle [\psi](u)|h(u)|[\psi](u)\rangle du\, .
\]
The connection is defined as follows: if $v \in S(W)$, then the horizontal distribution
at $v$ is defined by 
\[
{\cal H}_v = \{w\in T_v S(W)|\,  \langle v,w\rangle = 0\}\, .
\]
For a lifted path $v(t)\in S(W)$, the horizontality condition is
\[
\langle v(t),\dot v(t) \rangle = 0\, .
\]
The paths $\psi(t),\phi(t)$ differ by a phase, i.e., a smooth real-valued function $\theta(t)$ such that 
\[
\psi(t) = e^{-i\theta(t)} \phi(t),\quad \theta(t_0) = 0\, .
\]
Differentiate and multiply by $i$ to get
\[
i\dot \psi = \dot\theta e^{-i\theta} \phi + i e^{-i\theta}\dot \phi\, .
\]
Comparing this with the Schrödinger equation $i\dot\psi(t)=h(t)\psi(t)$
and cancelling $e^{-i\theta(t)}$ gives
\[
\dot\theta(t)\phi(t)+i\dot\phi(t)=h(t)\phi(t).
\]
Taking the Hermitian inner product with $\phi(t)$ in the second argument,
and using linearity in the first argument, we obtain
\[
\dot\theta(t)\langle\phi(t),\phi(t)\rangle
+i\langle\dot\phi(t),\phi(t)\rangle
=\langle h(t)\phi(t),\phi(t)\rangle.
\]
Since $\langle\phi(t),\phi(t)\rangle=1$ and
$\langle\dot\phi(t),\phi(t)\rangle=0$ by horizontality and Hermitian symmetry,
self-adjointness of $h(t)$ gives
\[
\dot\theta(t)=\langle h(t)\phi(t),\phi(t)\rangle
=\langle\phi(t),h(t)\phi(t)\rangle.
\]
Lastly, integrate to obtain
\[
\theta(t) = \int_{t_0}^t \langle h(u)\phi(u),\phi(u)\rangle\,du\, . \qedhere
\]
\end{proof}

\begin{rem} 
Let $\psi$ be a Floquet state with associated quasi-period $\lambda$.
Define the scalar $\mu\in U(1)$ by
\[
\Gamma([\psi],\psi(t_0))(t_0+1)=\mu\,\psi(t_0).
\]
Thus $\mu$ is the holonomy of the Fubini--Study connection around the
projective loop $[\psi]$. The lemma implies that
\begin{equation} \label{eqn:holonomy-lambda}
\lambda = \exp\left(-i \int_{t_0}^{t_0+1} \langle [\psi](u)|h(u)|[\psi](u)\rangle du\right)\,\mu\, .
\end{equation}
The scalar $\mu$ is independent of the choice of initial unit lift.
Each of the two scalar factors on the right of \eqref{eqn:holonomy-lambda}
is independent of the reference point $t_0$.

In physics terminology, write $\lambda=e^{i\gamma}$ and
$\mu=e^{i\gamma_g}$. Equation~\eqref{eqn:holonomy-lambda} gives
\[
\gamma\equiv\gamma_d+\gamma_g\pmod{2\pi},
\]
where
\[
\gamma_d:=-\int_{t_0}^{t_0+1}
\langle[\psi](u)|h(u)|[\psi](u)\rangle\,du
\]
is the {\it dynamical phase}. The phase $\gamma_g$, defined modulo
$2\pi$, is the {\it geometric phase}, also called the
{\it Aharonov--Anandan phase} of the cyclic evolution
\cite{Aharonov-Anandan}. This decomposition does not require an
adiabatic assumption.
\end{rem}
%\part{Stacks}

\subsection{Statement of the conjecture}
Let $h\:[0,1]\to\cal H_F(W)$ be a smooth family of good periodic Hamiltonians.
Choose a normalized Floquet state $\psi_0\in\cal S(W)$ of $h(0)$,
and fix an initial time $t_1\in\Bbb R$. Define
\[
M_s:=\mathcal U_{h(s)}(t_1+1,t_1),
\qquad L_0:=\Bbb C\psi_0(t_1).
\]
The operators $M_s$ depend smoothly on $s$ and have simple spectrum.
Each eigenvalue therefore admits a smooth local rank-one spectral projection.
Starting with $L_0$, these local spectral projections continue uniquely along
$[0,1]$, giving a smooth family of eigenlines $L_s$ and a smooth eigenvalue
branch $\lambda(s)\in U(1)$ such that
\[
M_s|_{L_s}=\lambda(s)\,\operatorname{id}_{L_s}.
\]
No eigenvalues collide, so the local continuations agree wherever they overlap.

Locally in $s$, choose a smooth unit vector $v_s\in L_s$ and set
\[
\psi_s(u):=\mathcal U_{h(s)}(u,t_1)v_s,
\qquad u\in\Bbb R.
\]
This solves the frozen Schr\"odinger equation for $h(s)$ with initial
condition $\psi_s(t_1)=v_s$. Periodicity of $h(s)$ gives
\[
\psi_s(u+1)=\lambda(s)\psi_s(u).
\]
Changing $v_s$ multiplies $\psi_s$ by a unimodular scalar independent of $u$.
Thus the resulting projective quasi-periodic states agree on overlaps and
define a smooth path
\[
[\psi]_\bullet\:[0,1]\longrightarrow\cal P(W),
\qquad [\psi]_0=[\psi_0].
\]
In particular, no choice of phase for $v_s$ is part of the conjecture's data.

The pair $([\psi]_\bullet,\psi_0)$ defines a lifting problem for the
principal $U(1)$-bundle
\[
p\:\cal S(W)\longrightarrow\cal P(W).
\]
The initial state $\psi_0$ is determined by $\psi_0(t_1)$ through the
frozen equation for $h(0)$, so we also write the lifting data as
$([\psi]_\bullet,\psi_0(t_1))$.
Below we construct a 1-parameter family of connections on $p$.

In the meantime, we consider evolution on $[t_1,t_2+N]$, with $t_1,t_2$
and $N$ as in the introduction. Set
\[
\tau(t):=\frac{t-t_1}{N+t_2-t_1}\, ,
\qquad t\in[t_1,t_2+N].
\]
The time rescaled Hamiltonian is
\[
\tilde h(t):=h(\tau(t))(t)\, .
\]
Let $\tilde\psi\:[t_1,t_2+N]\to S(W)$ solve the Schr\"odinger equation
for $\tilde h$ with initial condition $\tilde\psi(t_1)=\psi_0(t_1)$.

\begin{conjecture}[Adiabatic Floquet Conjecture] \label{conj:adiabatic-Floquet-first}
With $Q_N(t)$ and $P_s(t)$ as in the introduction, one has
\[
\sup_{t\in[t_1,t_2+N]}
\bigl\|Q_N(t)-P_{\tau(t)}(t)\bigr\|_{\operatorname{op}}
=O(1/N)
\]
as $N\to\infty$.
\end{conjecture}

\begin{rem} The Aharonov--Anandan identity suggests the formal approximation
\begin{equation} \label{eqn:asymptotic}
\tilde\psi(t) \approx \exp\left(-i \int_{t_1}^t \langle [\psi'](u)|\tilde h(u)|[\psi'](u)\rangle du\right)\,  \Gamma([\psi'],\tilde\psi(t_1))(t)\, ,
\end{equation}
where $[\psi']\:[t_1,N+t_2]\to\Bbb P(W)$ is defined by
\[
[\psi'](t)=[\psi]_{\tau(t)}(t)\, .
\]
A justification of this formal substitution requires estimates on the dynamical
integral and parallel transport over an interval of length $O(N)$,
beyond the projective estimate above. We defer these estimates to a subsequent paper.
\end{rem}

\subsubsection{The Floquet 1-form} The diffeomorphism 
\[
\cal P(W) \cong L\Bbb P(W) \times \Omega^1(S^1)
\]
of Proposition \ref{prop:projective-quasi-periodic} identifies the projection of
the principal $U(1)$-bundle
$p\: \cal S(W) \to \cal P(W)$
with the map
\[
\cal S(W)  \to L\Bbb P(W) \times \Omega^1(S^1)
\]
defined by $\psi\mapsto(\psi^\sharp,\omega)$, where $\omega$ is the second coordinate
of Proposition \ref{prop:projective-quasi-periodic}, obtained from the phase multiplier
using Lemma \ref{lem:qp-forms}.

We will associate to the latter  a  1-parameter family of diffeological 1-forms
\[
A^F\:\Bbb R \to \Omega^1(\cal P(W))\, .
\]
To accomplish this we must describe a rule which assigns to a plot
$\phi\: M \to \cal P(W)$ a family of $1$-forms $A^{F}_t(\phi) \in \Omega^1(M)$.
The rule is required to be compatible over all plots.
   
By Proposition \ref{prop:projective-quasi-periodic},
a  plot $M \to \cal P(W)$ corresponds to  a pair 
\[
(\psi^\sharp,\omega) \in C^{\infty}(M \times S^1 ,\Bbb P(W)) \times C^\infty(M,\Omega^1(S^1))    \, ,
\] 
in which
$\psi^\sharp \: M \times S^1 \to \Bbb P(W)$ is a family of free loops parametrized by $M$ and
 $\omega\: M \to \Omega^1(S^1)$ is a family of 1-forms on the circle.
 Let $\tilde \omega\: M \to \Omega^1(\Bbb R)$ denote the effect of pulling back
 forms along the universal cover. 
 
 Then
 \[
 \tilde \omega \in \Omega^0(M;\Omega^1(\Bbb R)) =: \Omega^{0,1}(M \times \Bbb R) \, .
 \]
Here $d\widetilde\omega$ denotes the exterior derivative on $M\times\Bbb R$.
Since $\widetilde\omega$ has only a component in the $\Bbb R$-direction,
$d\widetilde\omega$ has bidegree $(1,1)$.

We use real-valued connection forms, with the Fubini--Study connection
on $S(W)$ normalized by
\[
\mathcal A_v(w)=i\langle v,w\rangle=-i\langle w,v\rangle
\]
for $v\in S(W)$ and $w\in T_vS(W)$.
For a normalized representative $\psi_x$ of the state at $x\in M$, the
definition of $\omega$ gives
\[
\psi_x(u)=\psi_{x,0}(u)
\exp\left(-2\pi i\int_0^u\widetilde\omega_x\right),
\]
where $\psi_{x,0}$ is the Fubini--Study horizontal lift with
$\psi_{x,0}(0)=\psi_x(0)$. Consequently,
\[
i\langle\psi_x(u),\partial_u\psi_x(u)\rangle
=-2\pi\widetilde\omega_x(\partial_u).
\]

Let $\kappa$ be the pullback to $M\times\Bbb R$ of the real curvature
$F(\nabla)\in\Omega^2(\Bbb P(W);\Bbb R)$ in this normalization. Define
\[
A^F_t(\psi^\sharp,\omega)
=\int_t^{t+1}\left(u-t-\frac12\right)
\bigl(-2\pi\,d\widetilde\omega-\kappa\bigr)
\in\Omega^1(M).
\]
Integration is along the fibers of $M\times\Bbb R\to M$, with the convention
\[
\int_t^{t+1}\eta(u)\wedge du=\int_t^{t+1}\eta(u)\,du
\]
for a smooth family $\eta(u)\in\Omega^1(M)$; only the $(1,1)$-component
of the integrand contributes. This defines $A^F_t$ on the plot
$(\psi^\sharp,\omega)$. Compatibility with changes of plot follows from
naturality of exterior differentiation, curvature pullback, and fiber integration.

\subsubsection{The Floquet connection} For $t\in \Bbb R$ let $\nabla_t$ denote the connection
on the principal $U(1)$-bundle
\[
p\: \cal S(W) \to \cal P(W)
\]
defined by pulling back the Fubini-Study connection along 
the map
\[
\cal P(W) @>\epsilon_t >> \Bbb P(W)
\]
that evaluates a projective quasi-periodic state at time $t \in \Bbb R$.

Then
\[
\nabla^F_t = -A^F_t + \nabla_t, \qquad t\in\Bbb R
\]
defines a 1-parameter family of connections, called the {\it Floquet connection}.

In what follows we let
\[
\Gamma^F_t
\]
denote the associated 1-parameter family of lifting functions for the diffeological principal $U(1)$-bundle $p$.

\subsubsection{Auxiliary functions} The explicit expression for the adiabatic limit relies on 
two functions denoted by $\rho$ and $\varphi$.

We first define
\[
\rho\: C^\infty([0,1],U(1)) \to \Bbb R/(2\pi\Bbb Z)\, .
\]
Let $\lambda \: [0,1] \to U(1)$ be a smooth function. By choosing a branch of $\ln(\lambda(0))$ (i.e., a lift of the universal cover
$\Bbb R \to \Bbb R/2\pi \Bbb Z = U(1)$), we obtain a smooth function
\[
 -i\ln(\lambda) \: [0,1] \to \Bbb R\, .
\]
Set
 \[
 \rho (\lambda) = -i\int_0^1 \ln (\lambda(u))\,  du \pmod{2\pi}\, .
  \]
Changing the logarithm by $2\pi i k$, with $k\in\Bbb Z$, changes the real integral by $2\pi k$.
Hence its class modulo $2\pi$ is independent of the choice, as is $e^{iN\rho(\lambda)}$ for every integer $N$.

For $j = 1,2$,  define $\varphi_1,\varphi_2\: C^\infty([0,1],U(1)) \to \Bbb R$ by
\[
\varphi_j(\lambda) = i\left( \ln(\lambda(u_j))  - \int_0^1 \ln(\lambda(u))\, du  \right)\, , \qquad u_1 = 0, u_2 = 1\, ,
\]
and then define
\[
\varphi\: C^\infty([0,1],U(1)) \times \Bbb R^3  \to \Bbb R
\]
by 
\[
\varphi(\lambda, t_0,t_1,t_2) = (t_2-t_0)\varphi_2(\lambda) - (t_1-t_0)\varphi_1(\lambda) \, .
\]

\subsubsection{Geometric refinement} We are now in a position to state the geometric
refinement of Conjecture \ref{conj:adiabatic-Floquet-first}.
Retain the driving protocol $h\:[0,1]\to\cal H_F(W)$, the initial Floquet state
$\psi_0$, and the continued eigenvalue branch $\lambda$ constructed above.
The system starts at $\tilde\psi(t_1)=\psi_0(t_1)\in S(W)$, and
$\lambda(s)$ is the quasi-period of $[\psi]_s$.

\begin{conjecture}[Geometric Adiabatic Floquet Conjecture] \label{conj:adiabatic-Floquet-second-version}
As $N\to\infty$ through positive integers, we have
\begin{equation} \label{eqn:conj}
\tilde \psi(t_2+N) \approx e^{iN\rho(\lambda)} e^{i\varphi(\lambda,t_0,t_1,t_2)} \Gamma^F_{t_0}([\psi]_\bullet,\psi_0(t_1))(t_2)\, .
\end{equation}
Here ``$\approx$'' has the meaning given in Conjecture \ref{bigconj:geometric-adiabatic-Floquet}.
\end{conjecture}

\subsubsection{Dependence on the parameter $t_0$}
Even though each of the factors $e^{i\varphi(\lambda,t_0,t_1,t_2)}$ and  $\Gamma^F_{t_0}([\psi]_\bullet,\psi_0(t_1))(t_2)$
depends on $t_0$, we will see that their product does not. Let
\[
\alpha \in \Omega^1(U(1))\, ,
\]
be the canonical invariant 1-form
\[
-i\bar z dz, \qquad \int_{U(1)} \alpha =2\pi\, .
\]
The proof that the product of the factors does not depend on $t_0$ is then a direct consequence of the 
following result.

\begin{lem} The Floquet connection $\nabla_t^F$ and the auxiliary function $\varphi$ satisfy the identities
\begin{enumerate}[(i).]
\item $\nabla^F_{t_0+t}=\nabla^F_{t_0}+t\,\lambda^*\alpha$, where
$\lambda\:\cal P(W)\to U(1)$ is the quasi-period map and
$\lambda^*\alpha\in\Omega^1(\cal P(W))$.
\item For a smooth path $\lambda\:[0,1]\to U(1)$,
\[
\begin{aligned}
\varphi(\lambda,t_0+t,t_1,t_2)
&=\varphi(\lambda,t_0,t_1,t_2)\\
&\quad+t\int_{[0,1]}\lambda^*\alpha.
\end{aligned}
\]
\end{enumerate}
\end{lem}

\begin{proof}[Proof of (i):] Since the assertion is an equality of
diffeological connections, it is enough to prove it after pulling back
along an arbitrary plot
\[
q\colon M\longrightarrow\cal P(W).
\]
Represent $q$ by $(\psi^\sharp,\omega)$, and let
\[
E\longrightarrow M\times\mathbb R
\]
be the pullback of the canonical bundle
$S(W)\to\mathbb P(W)$. Let $\nabla$ be the pulled-back
Fubini--Study connection and set
\[
\bar\nabla=\nabla+2\pi\widetilde\omega.
\]
Its curvature is
\[
\bar F=\kappa+2\pi\,d\widetilde\omega.
\]

After restricting to an open set $U\subset M$, choose a smooth lift
of $q|_U$ to $\cal S(W)$. Evaluating these normalized quasi-periodic
states defines a section
\[
\Psi\colon U\times\mathbb R\longrightarrow E.
\]
Write $\lambda\colon U\to U(1)$ for their quasi-period, namely the
pullback of the quasi-period map along $q|_U$. Thus
\[
\Psi(x,u+1)=\Psi(x,u)\cdot\lambda(x).
\]
By the normalization of $\omega$, the connection form of $\nabla$
in this section has $du$-component $-2\pi\widetilde\omega$.
Hence $\Psi$ is $\bar\nabla$-horizontal in the $\mathbb R$-direction.
In particular, the monodromy of $\bar\nabla$ is the quasi-period
of the given state.

For fixed $u$, write $\Psi_u(x)=\Psi(x,u)$, and let
\[
a(u)\in\Omega^1(U)
\]
be the local connection form of $\nabla_u$ with respect to the section
$\Psi_u$. Since $\widetilde\omega$ has only an $\mathbb R$-component,
the restrictions of $\nabla$ and $\bar\nabla$ to
$U\times\{u\}$ agree. Moreover, $\Psi$ is
$\bar\nabla$-horizontal in the $\mathbb R$-direction. Thus the local
connection form of $\bar\nabla$ with respect to $\Psi$ is simply the
smooth family $a(u)$, with no $du$-component.

Since $U(1)$ is abelian, the curvature of $\bar\nabla$ in this
trivialization is
\[
\bar F
=
d_Ua(u)-\frac{\partial a}{\partial u}(u)\wedge du.
\]
Consequently,
\[
\bigl(-2\pi\,d\widetilde\omega-\kappa\bigr)^{1,1}
=
\frac{\partial a}{\partial u}(u)\wedge du.
\]
Only this $(1,1)$-component contributes to integration along the
$\mathbb R$-factor.

Set
\[
\beta=\lambda^*\alpha.
\]
The relation $\Psi_{u+1}=\Psi_u\cdot\lambda$ and the gauge
transformation rule for a $U(1)$-connection give
\[
a(u+1)=a(u)+\beta.
\]
It follows from the definition of the Floquet form that
\begin{align*}
A_t^F
&=
\int_t^{t+1}
\left(u-t-\frac12\right)
\frac{\partial a}{\partial u}(u)\,du\\
&=
\frac12a(t+1)+\frac12a(t)
-\int_t^{t+1}a(u)\,du\\
&=
a(t)+\frac12\beta-\int_t^{t+1}a(u)\,du.
\end{align*}
Therefore the local connection form of
$\nabla_t^F=\nabla_t-A_t^F$ is
\[
a^F(t)
=
\int_t^{t+1}a(u)\,du-\frac12\beta.
\]

Let
\[
J(r)=\int_r^{r+1}a(u)\,du.
\]
Using $a(r+1)-a(r)=\beta$, we obtain
\[
J'(r)=a(r+1)-a(r)=\beta.
\]
Hence
\[
a^F(t_0+t)-a^F(t_0)
=
J(t_0+t)-J(t_0)
=
t\beta
=
t\,\lambda^*\alpha.
\]
Since the plot $q$ and the open set $U$ were arbitrary, the local
identities assemble to give
\[
\nabla^F_{t_0+t}
=
\nabla^F_{t_0}+t\,\lambda^*\alpha.
\]
{\it Proof of (ii):} By direct calculation,
\begin{align*}
\varphi(\lambda,t_0+t,t_1,t_2)
&=(t_2-t_0-t)\varphi_2(\lambda)
  -(t_1-t_0-t)\varphi_1(\lambda)\\
&=\varphi(\lambda,t_0,t_1,t_2)
  +t\bigl(\varphi_1(\lambda)-\varphi_2(\lambda)\bigr).
\end{align*}
For a continuous choice of logarithm along $\lambda$, we have
\[
\varphi_1(\lambda)-\varphi_2(\lambda)
=i\bigl(\ln\lambda(0)-\ln\lambda(1)\bigr)
=\int_{[0,1]}\lambda^*\alpha.
\]
Substitution gives the asserted identity. \qedhere
\end{proof}

\appendix

\section{The Picard group}  \label{sec:Picard}
% Let $G = U(1)$  be the Lie group of unimodular complex numbers.
Let  $\xi\: L\to X$ be a smooth complex line bundle over a Fréchet manifold $X$.
Then $\xi$ determines a locally trivial smooth principal $\Bbb G_m$-bundle $\pi\:E\to X$.
Here, $\Bbb G_m = \Bbb C^*$ is the
multiplicative group of the nonzero complex numbers, and $E$ is the space of pairs $(\phi,x)$ in which $x\in X$
and  $\phi\: \Bbb C\to \xi ^{-1}(x)$ is an isomorphism
of complex vector spaces.
Conversely, a  smooth principal $\Bbb G_m$-bundle $\pi\: E\to X$ has an associated smooth complex line bundle
$E\times_{\Bbb G_m} \Bbb C\to X$.

If $G$ is a Lie group,  we let $\ubar G := C^\infty({-},G)$
be the sheaf of smooth $G$-valued functions.
 Recall that an {\it invertible sheaf} $\cal F$ on $X$ is a sheaf of $\ubar {\Bbb C}$-modules which is locally free of rank one.
 If  $\xi\: L\to X$ is a complex line bundle, then its sheaf of smooth local sections
is invertible. Conversely, given an invertible sheaf of $\ubar {\Bbb C}$-modules, there is
a unique smooth line bundle inducing it up to isomorphism. 

Summarizing, one has bijections between the following sets of isomorphism classes of objects over $X$:
{\small \begin{equation}\label{eqn:identifications}
\text{Principal $\Bbb G_m$-bundles} \quad \Leftrightarrow \quad \text{Line bundles} 
\quad \Leftrightarrow \quad \text{Invertible sheaves}\, .
\end{equation}}

% In view of the diffeomorphism \eqref{eqn:line-bundle-diff},
% we may identify such a line bundle with a stack morphism \[\xi\: \cal X\to [\ast/\Bbb G_m]\, .\]
% Recall that $[\ast/\Bbb G_m]$ has an atlas $\ast \to [\ast/\Bbb G_m]$ which is
% defined by the principal $\Bbb G_m$-bundle $\Bbb G_m\to \ast$. 
% Let  $\cal E(\xi) := \cal X \times_{[\ast/\Bbb G_m]}\ast$.
% Then the first factor projection 
% \[
%\pi\: \cal E(\xi) \to \cal X,
% \]
% (i.e., the $\Bbb G_m$-bundle projection), is  a surjective submersion. The stack $\cal E(\xi)$ is therefore the  ``total space'' of the bundle.

Given invertible sheaves $\cal F_1$ and $\cal F_2$ on $X$, the tensor product
operation
\[
\cal F_1 \otimes_{\ubar {\Bbb C}} \cal F_2
\]
is again an invertible sheaf.

The {\it Picard group}
\[
\Pic(X)
\]
is the abelian group of isomorphism classes 
of invertible sheaves on $X$, in which the group operation is induced
by tensor product and the group inverse is induced
by mapping an invertible sheaf $\cal F$ to its linear dual $\cal F^{\vee} := \underline{\hom}_{\ubar {\Bbb C}}(\cal F,\ubar {\Bbb C})$,
where $\underline{\hom}$ denotes internal $\hom$. By \eqref{eqn:identifications}, the Picard group may be 
alternatively defined as the abelian group of isomorphism classes of line bundles over $X$ with respect to the internal tensor product
operation on line bundles.
% Recall that $\cal G$ is the sheaf of smooth functions with values in $G$.

The following result is classical \cite[III, exer.~4.5]{Hartshorne}, \cite[prop.~3.4.6]{Milne}.

\begin{thm} \label{thm:Picard} There is a natural isomorphism of abelian groups
\[
\Pic(X) \cong H^1(X;\ubar {\Bbb G}_m)\, .
\]
\end{thm}

\subsection{Hermitian line bundles} A smooth complex line bundle $\xi\:L\to X$ is {\it Hermitian} if
it is equipped with a smooth complex-linear bundle map
\[
L\otimes_{\Bbb C}\overline L\longrightarrow X\times\Bbb C
\]
whose associated sesquilinear pairing on each fiber is Hermitian and positive definite.
Here $\overline L$ denotes the conjugate line bundle; our inner products are
linear in the first argument and conjugate-linear in the second.
Equivalently, a Hermitian metric specifies a reduction of the principal
$\Bbb G_m$-bundle of complex frames of $\xi$ to $S^1=U(1)\subset\Bbb G_m$.
 
Let $\hPic(X)$ denote the abelian group of isomorphism classes of Hermitian line bundles
over $X$.

\begin{add} Assume in addition that $X$ is nice. Then there is an isomorphism of abelian groups
\[
\hPic(X) \cong H^1(X;\ubar{S}^1)\, .
 \]
 \end{add}
 
 \begin{rem} If $X$ is not paracompact Hausdorff, then the homomorphism  $\hPic(X)\to \Pic(X)$ 
need not be an isomorphism, since there is no guarantee that every complex line
 bundle over $X$ admits a fiberwise Hermitian inner product.
 \end{rem}

\subsection{The first Chern class}
The  short exact sequence of sheaves
\[
0 \to \ubar{\Bbb Z} \to \ubar{\Bbb C} @> \exp(2\pi i{-}) >> \ubar{\Bbb G}_m \to 0
\]
induces a long exact sequence in sheaf cohomology
{\small \begin{equation} \label{eqn:exponential}
\cdots \to H^\ast(X; \ubar{\Bbb Z})\to  H^\ast(X; \ubar{\Bbb C}) \to H^\ast(X;\ubar{\Bbb G}_m) @> \delta >>   H^{\ast+1}(X; \ubar{\Bbb Z}) \to \cdots\, .
\end{equation}}
The  connecting homomorphism $\delta$
corresponds to the first Chern class
\[
c_1\: \Pic(X)  \to  H^2(X; \ubar{\Bbb Z})\, .
\]

\begin{lem} \label{lem:c-one} Let $X$ be a nice Fr\'echet manifold. Then
 $c_1$ is an isomorphism.
\end{lem}

\begin{proof} 
Since $X$ is nice, smooth partitions of unity exist \cite[Corollary~11.17]{Kihara}.
Hence the sheaf $\ubar{\Bbb C}$ is fine and therefore acyclic.
By the long exact sequence \eqref{eqn:exponential}, the connecting homomorphism
$\delta\: H^1(X;\ubar{\Bbb G}_m) \to H^2(X;\ubar{\Bbb Z})$ is an isomorphism.
The result follows. 
\end{proof}

\begin{rem} If $X$ is nice, then the homomorphism
\[
\hPic(X) \to \Pic(X)
\]
is an isomorphism. It is surjective by local triviality and the existence of
smooth partitions of unity. For injectivity, two Hermitian metrics $h_0,h_1$ on the same line bundle
satisfy $h_1=rh_0$ for a smooth positive function $r$ on $X$.
Multiplication by $r^{-1/2}$ is an isometry from $(L,h_0)$ to $(L,h_1)$.
\end{rem}

\section{Principal \texorpdfstring{$(T \rtimes H)$}{(T semidirect H)}-bundles} \label{sec:prin-TtimesH}
Let $G$ be a Lie group of the form $T\rtimes H$, where $T$ is an $n$-torus and $H$ is discrete.
Set $\Lambda = \pi_1(T)$. Then 
one has the induced action of $H$ on $\Lambda$. Let $\sigma\: H \to G$ be a homomorphism splitting the projection
$G\to H$, and let $i\: T\to G$ denote the inclusion.

 The goal of this section is to describe the isomorphism classes of smooth principal $G$-bundles
 on a nice Fréchet manifold $X$ in terms of the isomorphism classes of smooth $H$-principal bundles on $X$ and 
auxiliary cohomological data.

Given an $H$-principal
bundle $\xi\: E \to X$, we may form the associated covering space 
\[
E \times_H \Lambda \to X\, .
\]
Let 
\[
\cal L = \cal L(\xi,\Lambda)
\] 
denote the associated locally constant sheaf on $X$ with stalk $\Lambda$.

Consider pairs
\[
(\xi,\theta)\, ,
\]
where
\begin{enumerate}[(i).]
\item $\xi\: E\to X$ is a smooth principal $H$-bundle, and
\item $\theta \in H^2(X,\cal L)$
is a cohomology class. 
\end{enumerate}
Write
\[
(\xi,\theta) \sim (\xi',\theta')
\]
if there is an isomorphism of bundles $f\: \xi @>\cong >> \xi'$
covering the identity map of $X$ such that $f_\ast(\theta) = \theta'$,
where 
$f_\ast\: H^2(X,\cal L(\xi,\Lambda)) \to H^2(X,\cal L(\xi',\Lambda))$ 
is the homomorphism induced by the evident map of locally constant sheaves
$\cal L(\xi,\Lambda) \to \cal L(\xi',\Lambda)$. Then `$\sim$' defines an equivalence relation. Let 
\[
\Prin_H(X;\Lambda)
\]
denote the set of equivalence classes. Then $\Prin_H(X;\Lambda)$ has a preferred basepoint.
Let $\Prin_G(X)$ be the pointed set of isomorphism classes of smooth principal $G$-bundles on $X$.

\begin{prop} \label{prop:classify} If $X$ is a nice Fréchet manifold, then there is an isomorphism of pointed sets
\[
\Prin_G(X) \cong \Prin_H(X;\Lambda)\, .
\]
\end{prop}

\begin{proof}
Fix a principal $H$-bundle $\xi\:E\to X$, and write
$\cal L_\xi=\cal L(\xi,\Lambda)$. Let $\cal V_\xi$ and $\cal T_\xi$
be the sheaves of smooth sections of the associated bundles
$E\times_H\mathfrak t$ and $E\times_H T$, respectively, where
$\mathfrak t$ is the Lie algebra of $T$. The $H$-equivariant exponential
sequence $0\to\Lambda\to\mathfrak t\to T\to0$ induces a short exact
sequence of sheaves of abelian groups
\[
0\longrightarrow\cal L_\xi\longrightarrow\cal V_\xi
\longrightarrow\cal T_\xi\longrightarrow0.
\]
Since $X$ is nice, smooth partitions of unity exist
\cite[Corollary~11.17]{Kihara}. Thus $\cal V_\xi$ is fine
and hence acyclic. The connecting homomorphism is therefore an isomorphism
\[
\delta_\xi\:H^1(X;\cal T_\xi)\xrightarrow{\cong}H^2(X;\cal L_\xi).
\]

We first classify lifts of $\xi$, meaning principal $G$-bundles
$P\to X$ equipped with an isomorphism $P/T\cong E$, with isomorphisms
of lifts required to preserve this identification.
Choose local sections $s_i$ of $E$ over an open cover $\{U_i\}$, and
write $s_j=s_i h_{ij}$ on $U_{ij}=U_i\cap U_j$. After refining the
cover if necessary, a lift has transition functions of the form
\[
g_{ij}=i(a_{ij})\sigma(h_{ij}),\qquad a_{ij}\:U_{ij}\to T.
\]
Write $h\cdot a$ for the action of $H$ on $T$, so that
$i(h\cdot a)=\sigma(h)i(a)\sigma(h)^{-1}$.
The cocycle identity $g_{ij}g_{jk}=g_{ik}$ is equivalent to
\[
a_{ij}\bigl(h_{ij}\cdot a_{jk}\bigr)=a_{ik}.
\]
This is precisely the \v{C}ech cocycle condition for $\cal T_\xi$,
expressed in the local trivializations determined by the $s_i$.
Changing the local sections of $P$ by maps $b_i\:U_i\to T$, while
keeping their projections $s_i$ fixed, replaces $a_{ij}$ by
\[
a'_{ij}=b_i^{-1}a_{ij}\bigl(h_{ij}\cdot b_j\bigr).
\]
These are exactly the changes by a coboundary in $\cal T_\xi$.
Conversely, any such twisted cocycle defines a lift using the
transition functions $i(a_{ij})\sigma(h_{ij})$.
Passing to refinements, lifts of $\xi$ are therefore classified by
$H^1(X;\cal T_\xi)$, and hence by $H^2(X;\cal L_\xi)$.
The zero class corresponds to the lift with transition functions
$\sigma(h_{ij})$.

An isomorphism $f\:\xi\to\xi'$ induces isomorphisms of the three
associated sheaves and commutes with the connecting homomorphisms.
The lift defined by the fixed splitting $\sigma$ is also natural under
$f$. Thus transporting a lift along $f$ carries its cohomology class
$\theta$ to $f_\ast\theta$. Forgetting the specified identification
$P/T\cong E$ therefore imposes exactly the equivalence relation on
pairs $(\xi,\theta)$ defining $\Prin_H(X;\Lambda)$.
This proves the asserted bijection. The trivial $G$-bundle corresponds
to the trivial $H$-bundle with zero class, so the bijection preserves
basepoints.
\end{proof}

\begin{rem} Ignoring the issue of smoothness, a quick proof of Proposition \ref{prop:classify}
can be sketched as follows. Regard the map
\[
BG\to BH
\]
as a fibration with fiber $BT=K(\Lambda,2)$ and section induced by the splitting $\sigma$.
The action of $\pi_1(BH)=H$ on $\pi_2(BT)=\Lambda$ is the given action of $H$ on $\Lambda$.
For a fixed classifying map $u\:X\to BH$ of $\xi$, the section supplies a reference lift.
Homotopy classes of lifts of $u$, with homotopies through lifts of this fixed map,
are then identified with $H^2(X;\cal L(\xi,\Lambda))$.
To classify the underlying $G$-bundles, one must further identify these classes under
isomorphisms of the principal $H$-bundles, as in the equivalence relation defining $\Prin_H(X;\Lambda)$.
\end{rem}

In the following we let $H$  act
on  the underlying pointed set $H^2(X;\ubar\Lambda)$ by means of its action on $\Lambda$. 

\begin{cor} \label{cor:1-conn-classification} Assume in addition that $X$ is $1$-connected. Then there is an isomorphism of pointed sets
\[
\Prin_G(X) \cong H^2(X;\ubar\Lambda)_{H}\, ,
\]
where on the right, coinvariants are taken in
the category of pointed sets.
\end{cor}

\begin{proof} 
Since $X$ is simply connected and $H$ is discrete, every
principal $H$-bundle over $X$ is trivial. Choosing a
trivialization identifies its associated locally constant
sheaf with the constant sheaf $\underline{\Lambda}$.

Every automorphism of the trivial principal $H$-bundle
$X\times H$ covering the identity of $X$ has the form
\[
(x,k)\longmapsto (x,a(x)k),
\]
where $a\colon X\to H$ is smooth. Since $X$ is connected
and $H$ is discrete, $a$ is constant. Thus the automorphism
group is $H$, and its action on
$H^2(X;\underline{\Lambda})$ is induced by the given action
of $H$ on $\Lambda$.

The equivalence relation on pairs $(\xi,\theta)$ therefore
identifies precisely the $H$-orbits of cohomology classes.
Proposition~\ref{prop:classify} now gives
\[
\operatorname{Bun}_G(X)
\cong H^2(X;\underline{\Lambda})_H.
\]
The trivial $G$-bundle corresponds to the orbit of zero.
\end{proof}

%\begin{thebibliography}
%\bibliographystyle{amsplain}
%\bibliography{periodic}

\end{document}